\documentclass[%
 reprint,
superscriptaddress,
nofootinbib,
nobibnotes,
 amsmath,amssymb,
 aps,floatfix
]{revtex4-2}

\usepackage[inline]{enumitem}
\usepackage{amsthm}
\usepackage{graphicx,xcolor}
\usepackage{dcolumn}
\usepackage{bm}
\usepackage[breaklinks=true,colorlinks,citecolor=blue,linkcolor=blue,urlcolor=blue]{hyperref}
\usepackage[mathlines]{lineno}
\usepackage[caption=false, position=bottom]{subfig}
\usepackage{floatrow}
\usepackage{mathtools}
\usepackage{mwe}
\usepackage{hyperref}
\usepackage{float}
\usepackage{booktabs}
\usepackage{multirow}

\newcommand\norm[1]{\| #1 \|}
\usepackage[utf8]{inputenc}
\usepackage[english]{babel}
\usepackage{theoremref}
\newtheorem{theorem}{Theorem}

\begin{document}

\title{Entanglement Diagnostics for Efficient Quantum Computation}
\author{Joonho Kim}
\affiliation{
 School of Natural Sciences, Institute for Advanced Study, Princeton, NJ 08540, USA.
}
 
\author{Yaron Oz}%
\affiliation{
 School of Natural Sciences, Institute for Advanced Study, Princeton, NJ 08540, USA.
}
\affiliation{
 Raymond and Beverly Sackler School of Physics and Astronomy, Tel-Aviv University, Tel-Aviv 69978, Israel.
}

\begin{abstract}

We consider information spreading measures in randomly initialized variational quantum circuits and introduce entanglement diagnostics for efficient variational quantum/classical computations. 
We establish a robust connection between entanglement measures and optimization accuracy by solving two eigensolver problems for Ising Hamiltonians with nearest-neighbor and long-range spin interactions. 
As the circuit depth affects the average entanglement of random circuit states, the entanglement diagnostics can identify a high-performing depth range for optimization tasks encoded in local Hamiltonians. 
We argue, based on an eigensolver problem for the Sachdev-Ye-Kitaev model, that entanglement alone is insufficient as a diagnostic to the approximation of volume-law entangled target states and
that a large number of circuit parameters is needed for such an optimization task.

\end{abstract}

\maketitle

\section{Introduction}

Noisy Intermediate-Scale Quantum (NISQ) technology is being developed rapidly and poses a great challenge to come up with efficient quantum algorithms \cite{nisq}, which will operate on the NISQ computers and perform better than classical algorithms. Many real-world use cases are associated with machine learning and optimization, for which variational quantum circuits offer an appropriate framework. 
The typical optimization tasks can be formulated as a search for the ground state of a 
Hamiltonian $H$, which may encode an exact combinatorial problem \cite{QAOA,Lucas_2014}.

The variational quantum algorithms (VQA) consist of two elements \cite{Peruzzo_2014}. The first part is quantum, where one constructs a parameterized quantum circuit composed of $L$ unitary layers on the product state of $n$ qubits, $|0\rangle^{\otimes n}$. 
The layer unitaries and quantum gates therein depend on continuous parameters,  each initialized with the uniform measure on $[0, 2\pi)$. Denoting all the circuit parameters collectively by $\theta$, 
the variational state is written as 
\begin{equation}
|\psi_c(\theta)\rangle = U(\theta)|0\rangle^{\otimes n} \ .
\end{equation}
The second part of the  variational quantum algorithm is classical, where we estimate the Hamiltonian expectation value with the variational circuit state, i.e.,
\begin{equation}
E(\theta) = \langle \psi_c(\theta)| H  |\psi_c(\theta) \rangle \ ,
\label{E}
\end{equation}
and minimize it in the $nL$-dimensional parameter space using the gradient descent method.

\begin{figure}[t]
\centering
\subfloat[]{
        \centering
        \label{fig:circuit_diagram_left}
        \includegraphics[height=1.85cm]{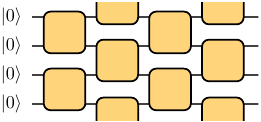}}
\hfill
\subfloat[]{
        \centering
        \label{fig:circuit_diagram_right}
        \includegraphics[height=1.85cm]{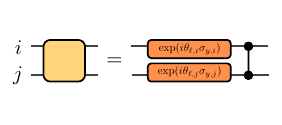}} 
\caption{The circuit architecture used in Sections~\ref{sec:randomcircuit}~and~\ref{sec:opt}. (a) The horizontal axis can be interpreted as the discrete time $L$. We call the commuting set of simultaneous 2-qubit gates as the circuit layer. (b) Each gate consists of the single-qubit Pauli-$y$ rotations \eqref{eq:paulirot} followed by the CZ operation \eqref{eq:cz}.}
\label{fig:circuit_diagram}
\end{figure}

Entanglement encodes information in the qubit correlations, which are generated by the successive application of the circuit layers. 
Given two complementary systems $A/B$, the Renyi-$k$ entropy of the reduced density matrix
\begin{equation}
{\cal R}^{k}_A = \frac{1}{1-k} \log \text{Tr} \left(\rho^k_A\right)
\label{Re}
\end{equation} 
measures their entanglement, so does the von Neumann entropy that corresponds to \eqref{Re} in the special limit $k\rightarrow 1$:
\begin{equation}
S_{EE} = - \text{Tr}\, \rho_A \log \rho_A  \ .
\label{SEE}
\end{equation}
The reduced density matrix $\rho_{A}$ 
is obtained from the full circuit density matrix $\rho_c(\theta) = |\psi_c(\theta)\rangle\langle \psi_c(\theta)|$ by taking a partial trace over the subsystem $B$.

The performance of the variational quantum algorithm depends largely on whether the quantum circuit can prepare an initial variational state $|\psi_c(\theta) \rangle$ that is close to the target ground state $|\psi_g\rangle$ of the Hamiltonian.
In this paper we argue that the average entanglement entropy \eqref{Re} or \eqref{SEE} of random circuit states provides
a distance measure that can quantify a successful minimization of the energy function. Note that, for their computation, 
we specifically use the equal partition $n_A=n_B=n/2$ and the binary logarithm.

\begin{figure}[t]
\centering
\includegraphics[width=6cm]{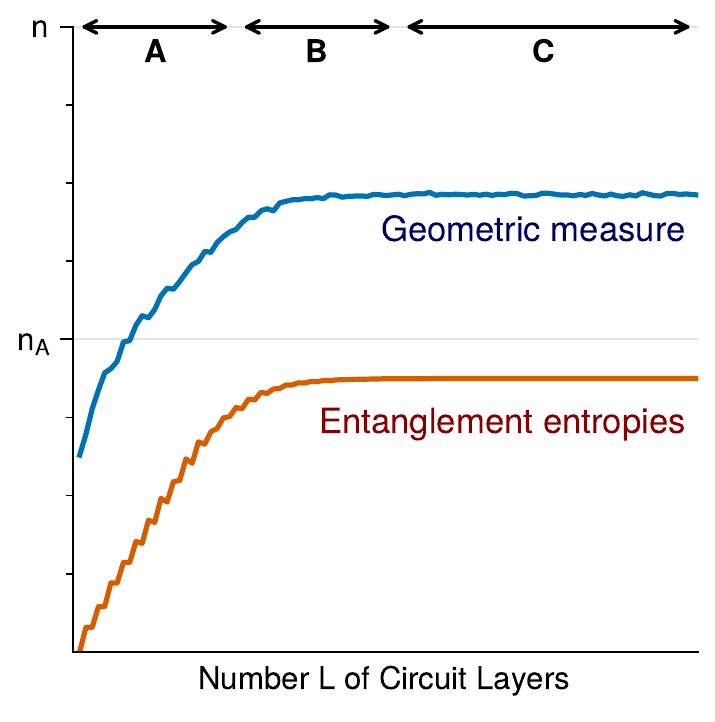}
\caption{
A schematic plot for the growth of mean entanglement entropies and geometric distance of circuit-generated states. We often distinguish the low, intermediate, high-depth circuits in Sections~\ref{sec:randomcircuit}~and~\ref{sec:opt} and denote their corresponding depth ranges by $A$, $B$, and $C$, respectively.
}
\label{fig:random_growth}
\end{figure}

The evolution of the entanglement entropies as a function of the circuit depth $L$ is schematically drawn in Figure~\ref{fig:random_growth}. It is convenient to divide the range of $L$ into three regions $A$, $B$, and $C$.\footnote{Such usage of $A$, $B$, $C$ should be distinguishable from the other usage of $A$, $B$ that denotes a subset of $n$ qubits.} $A$ is where the entanglement entropy continues to grow, while $C$ is where it has saturated to a constant value. As for their scaling behavior in $n$, the random circuit states in $A$/$C$ obey the area/volume law scaling of the entanglement entropies, respectively. Since the ground states of gapped local Hamiltonians  are expected to have an area law entanglement entropy, we expect that an initial variational state $|\psi_c(\theta) \rangle$ in $A$ would lead to efficient VQA optimization in contrast to those circuits in $C$. 
We also identify $B$ as a transition region between $A$ and $C$, where the entanglement entropy has already reached saturation yet the initial random parameter can determine the success/failure of the VQA optimization.

The technical reason why the circuit optimization fails in region $C$ is the vanishing gradient problem. When the circuit distribution is approximately 2-design, such that the first and second moments are indistinguishable from those of the Haar distribution, the energy gradient at initial random values cannot deviate from zero, except for an exponentially decaying probability in $n$ \cite{Mc,cost-dep-bp,entanglement-bp, holmes2021connecting}. It happens for the circuit ensemble in $B/C$, where the Renyi-$2$ entropy as a diagnostic of the quantum 2-design is closest to $n_A=n/2$ that corresponds to the Haar ensemble.

Until now, we assumed that the entanglement entropy of the target state follows the area law scaling, as in gapped local one-dimensional systems \cite{Hastings_2007}. However, it does not always hold, and the variational circuit in $A$ cannot minimize the circuit energy \eqref{E} to the ground level. For the Sachdev-Ye-Kitaev model \cite{SYK,2015Kitaev,Maldacena_2016}, whose ground state exhibits a volume law entanglement \cite{Huang_2019},  the optimization does fail no matter to which of $A/B/C$ the variational circuit belongs. Incidentally, we argue that higher-dimensional parameter space can assist the circuit optimization even at high level of circuit state entanglement, so that over-parameterized circuits can offer a high precision approximation of volume law entangled target states including the SYK ground state \cite{highdepth}.

The rest of this paper is organized as follows: Section~\ref{sec:ent} motivates the entanglement diagnostics as the initialization condition to arrange variational states close to the target. Section~\ref{sec:randomcircuit} studies the average entanglement growth of circuit states as a function of the circuit depth. Section~\ref{sec:opt} examines the importance of the entanglement diagnostics in the local gradient search of optimal circuit parameters. Section~\ref{sec:other} checks the validity of the entanglement diagnostics by testing them against different circuit architectures and also discuss the impact of shrinking the circuit parameter dimension. The paper concludes with discussion and outlook in Section~\ref{sec:dis}.
Additional details are given in the appendices.

\section{Entanglement Diagnostics}
\label{sec:ent}

Using the density matrix of the quantum circuit $\rho_c(\theta)$, the expectation value of the Hamiltonian (\ref{E}) reads:
\begin{equation}
E(\theta)  = \text{Tr} \left(\rho_c(\theta)  H\right) . \label{E1}
\end{equation}
Our optimization task is to get as close as possible to the ground state of the Hamiltonian by minimizing \eqref{E1}. 
It can be achieved by multiple iterations of evaluating the density matrix $\rho_{c}(\theta)$ and updating the parameters via the gradient descent \eqref{eq:gd} that will finally stop at $\theta = \theta_f$. 
We would like to reach the final parameter $\theta_{f}$ such that 
\begin{equation}
\label{eq:error}
\Delta E \equiv \text{Tr} \left((\rho_{c}(\theta_{f}) - \rho_{g})H \right) \simeq 0
\end{equation}
where $\rho_{g}$ is the exact ground state of the Hamiltonian.
A simple upper bound of this approximation error $\Delta E$ follows from the Cauchy-Schwarz inequality,
\begin{equation}
\text{Tr} \left((\rho_{c}(\theta_{f}) - \rho_{g})H \right) \leq
{\norm{\rho_{c}(\theta_{f}) - \rho_{g}}}_1 \cdot {\norm{H}}_1,
\label{BH}
\end{equation}
where the trace norm  $\norm{O}_1$ is the sum of singular values of an operator $O$, i.e., eigenvalues of $(O^{\dagger}O)^{1/2}$. 

A natural condition for efficient reduction of $\Delta E$ is arranging an initial circuit state $\rho_{c}(\theta_{in})$ to 
be in the proximity of the ground state with a small enough trace distance ${\norm{\rho_{c}(\theta_{in}) - \rho_{g}}}_1$. However, we will confront two issues.
First, we generally do not know the ground state, thus being unable to estimate the trace distance ${\norm{\rho_{c}(\theta_{in}) - \rho_{g}}}_1$. 
Second, the trace distance can be very sensitive to tiny changes of quantum states. So the above condition is often over-restrictive, discarding most reasonable initial states.

Instead, we want to relax the condition by using the entanglement entropy of an initial circuit state as a distance proxy between $\rho_c(\theta_{in})$ and $\rho_{g}$, from which one can expect the success/failure of circuit optimization.  It can be motivated as follows: The inequalities on the von Neumann and Renyi-$k$ entropy differences \cite{Fannes1973, raggioa1995properties,Audenaert_2007,Chen2016}:
\begin{align}
|S_{EE}(\rho_A) - S_{EE}(\sigma_A)|
\leq &
\ n\norm{\rho_A - \sigma_A}_1  + (e\ln{2})^{-1},\\
|\mathcal{R}^k(\rho_A) - \mathcal{R}^k(\sigma_A)|
\leq &
\ k\cdot  2^{n_A(k-1)}\norm{\rho_A - \sigma_A}_1,
\end{align}
show that, for given two quantum states $\rho_A$ and $\sigma_A$, being close in their quantum entropies is necessary for being close in their trace distance $\norm{\rho_A - \sigma_A}_1$. Suppose now that $\rho_A$ and $\sigma_A$ are the density matrices
of a subsystem $A$, which can be combined with a complementary part $B$ to constitute the entire $n$ qubit system, i.e., $\rho_A \equiv \text{Tr}_B \left(\rho\right)$ and $\sigma_A \equiv \text{Tr}_B \left(\sigma\right)$. Monotonicity of the trace distance under the partial trace, 
\begin{equation}
 {\norm{\rho_{A} - \sigma_{A}}}_1 \leq     {\norm{\rho - \sigma}}_1 \ ,
 \label{B1}
\end{equation}
implies in turn that, for the trace distance between two quantum states $\rho$ and $\sigma$ to be small, the difference in their entanglement entropies should necessarily be small. Hence, the entanglement diagnostics of initial circuit states can be considered as a weaker version of the proximity measure.

We usually cannot estimate the trace distance from the exact ground state $\rho_g$ due to our ignorance of $\rho_g$. However, we expect that the ground states of gapped local Hamiltonians are far from typical quantum states $\sigma$, whose $\sigma_A$ are approximately maximally mixed \cite{Goldstein_2006}. Thus, we require the trace distance between the equiprobable state and the reduced density matrix of the circuit state $\rho_A$ to be large. This requires non-maximal entanglement entropies of the circuit states, i.e. they should not scale with the subsystem size $n_A$ \cite{Page_1993}. This is encoded in the following:

\begin{theorem}

\noindent
(i) The trace distance between the reduced density matrix $\rho_A$ 
and the maximally mixed state ${2^{-n_A}}\,I_{n_A}$ satisfies the following inequality:
\begin{align}
 \mathcal{L}(n_A, \rho_A) \leq \frac{1}{2}\left\Vert\rho_A - \frac{I_A}{2^{n_A}}\right\Vert_1  \leq 
 \mathcal{U}(n_A, \rho_A) 
\end{align}
with
\begin{align}
 \mathcal{U}(n_A, \rho_A)  &= \left(\frac{n_A - S_{EE}(\rho_A)}{2}\right)^{1/2}\label{B4}\\
 \mathcal{L}(n_A, \rho_A)  &= \frac{1}{2k} \left(   2^{ (1-k)\mathcal{R}^k_A(\rho_A)} - 2^{(1-k)n_A}\right)
 \label{eq:lower_bound}
\end{align}
where $S_{EE}(\rho_A)$ and $\mathcal{R}^q_A(\rho_A)$ are the von Neumann and Renyi-$k$ entropies of the reduced state $\rho_A$, respectively.
\\[0.2cm]\noindent
(ii) In the large size limit $n_A \gg 1$ of the subsystem $A$, the following lower bound holds asymptotically:
\begin{equation}
  1 - \frac{S_{EE}\left(\rho_{A}\right)}{n_A}
\lesssim \frac{1}{2}\left\Vert{\rho_A - \frac{I_A}{2^{n_A}}}\right\Vert_1 \ .
\label{eq:lower_bound2}
\end{equation}
\end{theorem} 

\begin{proof}
\noindent
(i) 
We start from the Pinsker's inequality: 
\begin{equation}
 \tfrac{1}{2} {\norm{\rho - \sigma}}^2_1 \leq S(\rho ||\sigma)  
 \label{B0}
\end{equation}
on the trace distance between two states $\rho$ and $\sigma$ and their relative entropy. Plugging $\rho=\rho_A$ and $\sigma = 2^{-n_A}\,{I}_{A}$, 
\begin{align}
    S(\rho_A||\sigma_A) =  n_A + \text{Tr}(\rho_A \log \rho_A) = n_A - S_{EE}(\rho_A) 
\end{align}
such that (\ref{B0}) becomes (\ref{B4}).

The continuity bound of the Tsallis-$k$ entropy implies \cite{raggioa1995properties}:
\begin{equation}
|\text{Tr} \left(\rho_A^k\right) - \text{Tr} \left(\sigma_A^k\right)|\leq k \,\norm{\rho_A - \sigma_A}_1 \ ,
\label{ineqn}
\end{equation}
which can turn into \eqref{eq:lower_bound} by inserting $\sigma_A = 2^{-n_A}\,{I}_{A}$ and  $\text{Tr} \left(\rho_A^k\right) = 2^{(1-k)\mathcal{R}^k_A(\rho_A)}$. 
\\[0.2cm]\noindent
(ii)
Recall the Fannes–Audenaert inequality \cite{Audenaert_2007}:
\begin{align}
|S_{EE}(\rho_A) - S_{EE}(\sigma_A)|
\leq &
\tfrac{1}{2}\norm{\rho_A - \sigma_A}_1 \log (2^{n_A}-1) \nonumber\\
&+ H(\tfrac{1}{2}\norm{\rho_A - \sigma_A}_1) \label{eq:fannes}
\end{align}
where $H(t) \equiv -t \log{t} - (1-t)\log(1-t)$. Substituting $\sigma_A = 2^{-n_A}\,{I}_{A}$ and taking the large system size limit $n_A\gg 1$, the LHS of \eqref{eq:fannes} becomes $(n_A - S_{EE}(\rho_A))$, which leads to the asymptotic inequality \eqref{eq:lower_bound2}.
\end{proof}

We stress that the entanglement diagnostic for circuit states is only a necessary condition to keep the initial and target states close. Remarkably, as we will see in Section~\ref{sec:opt}, the gradient-based optimization indeed works efficiently for those variational circuits whose average entanglement entropy scales slower than the volume law.
Concerning the circuit depth, this suggests to avoid intermediate-depth and high-depth circuits, respectively corresponding to $B/C$ of Figure~\ref{fig:random_growth}, and favor the circuits with fewer layers that belong to $A$.
We will estimate the critical depth $L_s$ that divides $A$ and $B/C$ in the following Section~\ref{sec:randomcircuit}.

\section{Random Quantum Circuit}
\label{sec:randomcircuit}

In this section, we study the growth of entanglement entropy for the circuit states generated by a random circuit evolution of the initial product state $|0\rangle^{\otimes n}$. 
Figure~\ref{fig:circuit_diagram} is the quantum circuit architecture used in this paper. It defines a $(1+1)$-dimensional discrete quantum system, where the $n$ qubits along the vertical axis represent the space, and the $L$ layers along the horizontal axis span the time. The qubits are arranged identically with period $n$, i.e., $i\simeq i+n$, imposing a periodic boundary condition along the spatial direction. At each time step, the wavefunction evolves by a chain of the 2-qubit unitary gates, acting alternatingly on all neighboring odd-even/even-odd qubit pairs. The 2-qubit gate is made of independent Pauli-$y$ rotations acting on single qubits, 
\begin{equation}
R(\varphi) =  \exp(i  \sigma_y \varphi) =    \begin{pmatrix}
\cos\varphi & \sin\varphi \\
-\sin\varphi & \cos\varphi
\end{pmatrix} \ ,
\label{eq:paulirot}
\end{equation}
followed by the controlled-Z operation
\begin{align}
    CZ = \text{diag}\left(1, 1, 1, -1\right)
    \label{eq:cz}
\end{align}
that generically creates a pairwise entanglement. We will collectively denote all rotation angles by $\theta$ while using $\theta_{l,i}$ to indicate a specific angle that rotates the $i$'th qubit at the $l$'th layer, where $1\leq i \leq n$ and $1\leq \ell\leq L$. These variables are randomly chosen from $\mathcal{U}(0,2\pi)$, the uniform distribution  between $0$ to $2\pi$.

\begin{figure}[t]
\subfloat[Von Neumann Entropy]{
    \includegraphics[height=4.1cm]{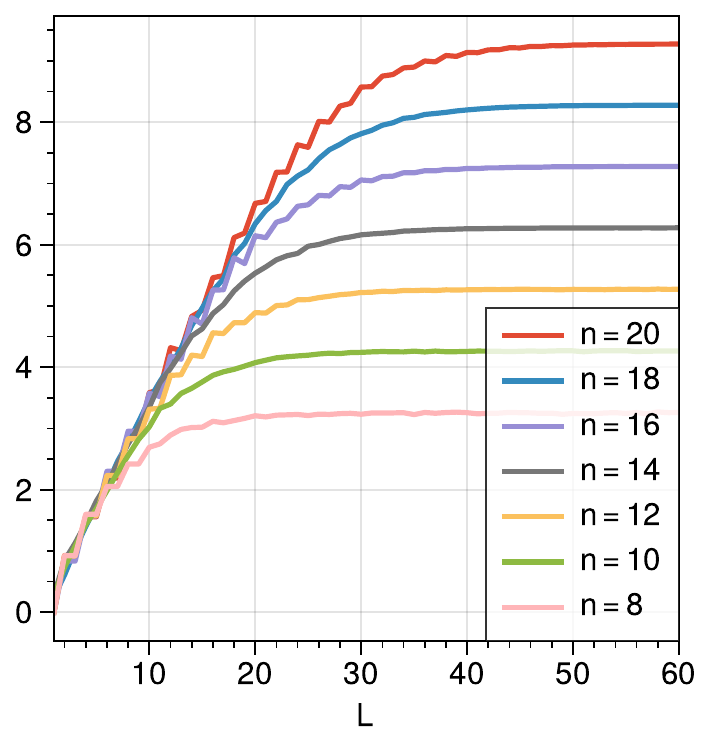}\label{fig:vn}}
\subfloat[Renyi-2 Entropy]{
    \includegraphics[height=4.1cm]{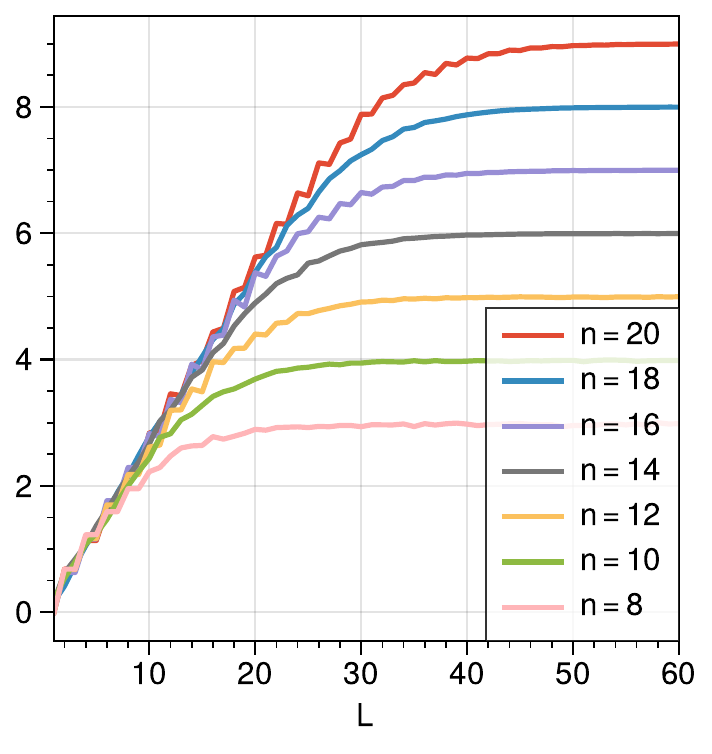}\label{fig:r2}}
\caption{Von Neumann and Renyi-2 entropies for $8 \leq n \leq 20$ averaged over 50 independent circuit states, as a function of the circuit depth $L$. }
\label{fig:entropy-growth-generic}
\end{figure}

\subsection{Linearity of the 
Initial Entanglement Growth}

Let us consider the evolution
of the $n$-qubit state under the random circuit unitaries of Figure~\ref{fig:circuit_diagram}, as a function of the number of layers $L$. 
We measure the average growth of the bipartite entanglement of random circuit states by decomposing the $n$ qubits into two equal-size subsystems, $n_A = n_B= n/2$, and calculating the sample statistics of various Renyi entropies for different $n$ and $L$.  

Figures~\ref{fig:vn}~and~\ref{fig:r2} show the von Neumann and Renyi-2 entropies averaged over $50$ random circuit states
with different numbers of qubits $n$. Figure~\ref{fig:entropies-generic} compares the Renyi entropies of different orders averaged over $50$ random circuit states with $n=12$ and $20$ qubits.
They all exhibit the linear growth of the entanglement entropies at initial times. The curves then slow down in growth and eventually reach the plateaus.
See Figures~\ref{fig:pmax}--\ref{fig:geometry-generic} in the appendix for the growth curves of several other entanglement quantities with different system sizes $n$.

The early linear growth of the entanglement entropy, 
\begin{equation}
 {\cal R}^k_{A} (\rho_A) = v_{k} L \ ,
 \label{eq:area-law}
\end{equation}
is a characteristic feature of the global quench dynamics \cite{Calabrese_2005}, which in our case is driven by the successive application of the layer unitaries $U(\theta_L)$ to the $n$-qubit product state $|0\rangle^{\otimes n}$. The coefficient $v_k$ is known as the entanglement velocity and  generally depends on $k$. 
We determine $v_k$ by the linear regression of the early-time entropies on the range of $0 \leq L \leq n/2$. The estimated values of $v_k$, computed at different $n$'s and $k$'s, are summarized in the third columns of Tables~\ref{tbl:random}~and~\ref{tbl:random-extra}. 
We find that
$v_k$ is independent of $n$ except for minor fluctuations, identifying \eqref{eq:area-law} with the area-law entanglement of the early-time circuit states.
Furthermore, $v_k$ decreases when the order $k$ of the Renyi entropy increases, i.e., $v_{k_1} > v_{k_2}$ for $k_1 < k_2$.

On the other end, at a late time, the Renyi-$k$ entropy saturates to a constant $r_{n,k}$ for any $n$ and $k$. We compute the saturated value of $\mathcal{R}^{k}_A$ by averaging it over the time frame $200 \leq L \leq 250$ and record that in the fifth columns of Tables~\ref{tbl:random}~and~\ref{tbl:random-extra}. The resulting constants $r_{n,k}$ manifest the following simple dependency on $n_A = n/2$:
 \begin{equation}
 {\cal R}^k_{A} (\rho_A) =  r_{n,k} \simeq n/2 - c_k = \text{Vol}(A) - c_k \  
 \label{eq:volume-law}
 \end{equation}
 implying the volume-law entanglement of the late-time circuit states \cite{Liu_2018}. We also find that, as the entropy order $k$ increases, the saturated value $r_{n,k}$ declines monotonically, so the shift constant $c_k>0$ can be only larger.

 \begin{figure}[t]
\subfloat[$n=12$]{
    \includegraphics[height=4.1cm]{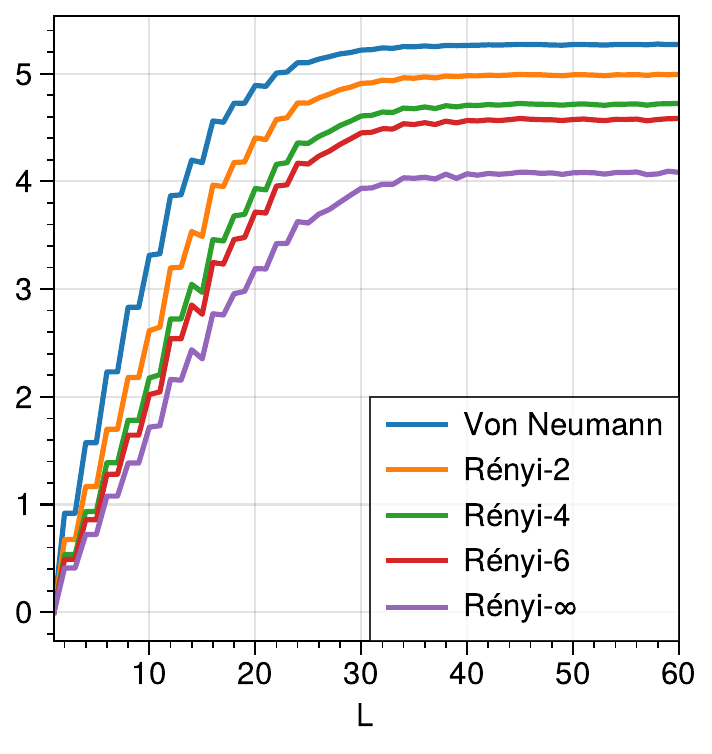}\label{fig:n12}}
\subfloat[$n=20$]{
    \includegraphics[height=4.1cm]{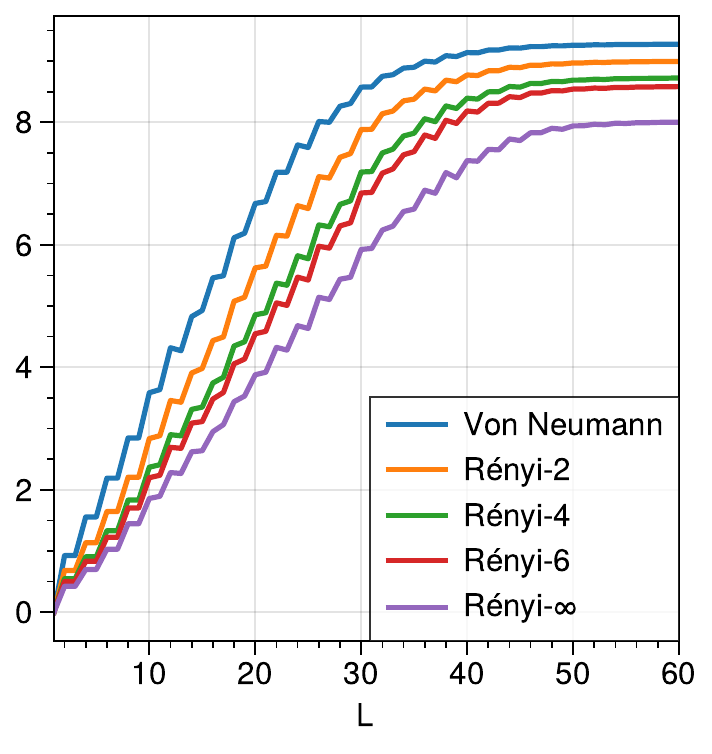}\label{fig:n20}}
\caption{Various Renyi-$k$ entropies for $n=12$ and $n=20$ averaged over 50 independent circuit states, as a function of the circuit depth $L$.}
\label{fig:entropies-generic}
\end{figure}
Combined with the discussion in Section~\ref{sec:ent}, the average entanglement curves suggest to refrain from using a variational circuit in the region of the plateau, i.e., $L \geq L_s$, in order to prepare an initial circuit state in proximity to the target ground state that follows the area-law entanglement. We now turn to examine the scaling behavior of the early-time and late-time scales, i.e., $L_l$ and $L_s$.

\begin{table}[b]
\centering
\begin{tabular}{@{}ccccccccccccc@{}}
\toprule
Type &  $n$ & $v_k$ & $L_l$ & $r_{n,k}$ & $L_s$ & Type &  $n$ & $v_k$ & $L_l$ & $r_{n,k}$ & $L_s$\\ \midrule
\multirow{7}{*}{$S_{EE}$} 
& $8$ & 0.3669 & 7 & 3.2526 & 29 & \multirow{7}{*}{$\mathcal{R}^2_A$} & $8$ & 0.2771 & 9 & 2.9722 & 29\\
& $10$ & 0.3480 & 6 & 4.2639 & 33 && $10$ & 0.2581 & 12 & 3.9821 & 36\\
& $12$ & 0.3533 & 13 & 5.2703 & 43 && $12$ & 0.2645 & 17 & 4.9896 & 43\\
& $14$ & 0.3551 & 9 & 6.2743 & 49 && $14$ & 0.2711 & 15 & 5.9944 & 49\\
& $16$ & 0.3522 & 17 & 7.2766 & 58 && $16$ & 0.2689 & 23 & 6.9973 & 60\\
& $18$ & 0.3459 & 12 & 8.2776 & 65 && $18$ & 0.2682 & 24 & 7.9986 & 65\\
& $20$ & 0.3411 & 23 & 9.2781 & 71 && $20$ & 0.2619 & 33 & 8.9993 & 71\\\bottomrule
\end{tabular}\qquad
\caption{
Phenomenological estimation of the entanglement velocity, saturation value, early-time scale, and late-time scale.}
\label{tbl:random}
\end{table}
\subsection{Timescale for the Entanglement Growth}

Let us study the early-time $L_l$ and late-time scales $L_s$, respectively, as the depth scales where the linear growth \eqref{eq:area-law} ends and where the saturation \eqref{eq:volume-law} begins. We measure $L_l$ and $L_s$ using the following operational definitions: 
 \begin{align}
     \textstyle L_l &= \max\left\{L : |\mathcal{R}^k_A(L) - v_k L| \leq  2\, \text{RMS}_{k}({0, n/2})\right\}
     \label{pheno}\\
     \textstyle L_s &= \min\left\{L : |\mathcal{R}^k_A(L) - r_{n,k}| \leq  2\,\text{RMS}_{k}({200, 250})\right\},\nonumber
 \end{align}
 $L_l$ is the maximum depth $L$ where the gap $|\mathcal{R}^k_A(L) - v_k L|$ between the Renyi entropy and its linear approximation maintains smaller than two times the RMS deviation,
 \begin{align}
     \textstyle\text{RMS}_{k}({a, b}) =  \sqrt{\frac{1}{b-a+1}  \sum_{\ell=a}^{b}\ (\mathcal{R}^k_A(\ell) - v_k \ell)^2},
     \label{eq:RMS}
 \end{align}
 for $0 \leq L \leq n/2$. Similarly, $L_s$ is the minimum depth $L$ whose difference $|\mathcal{R}^k_A(L) - r_{n,k}|$ between the Renyi entropy and its saturated value remains to be  smaller than two times the RMS deviation \eqref{eq:RMS} for $200 \leq L \leq 250$.

\begin{figure*}[t]
    \centering
    \includegraphics[width=0.85\textwidth]{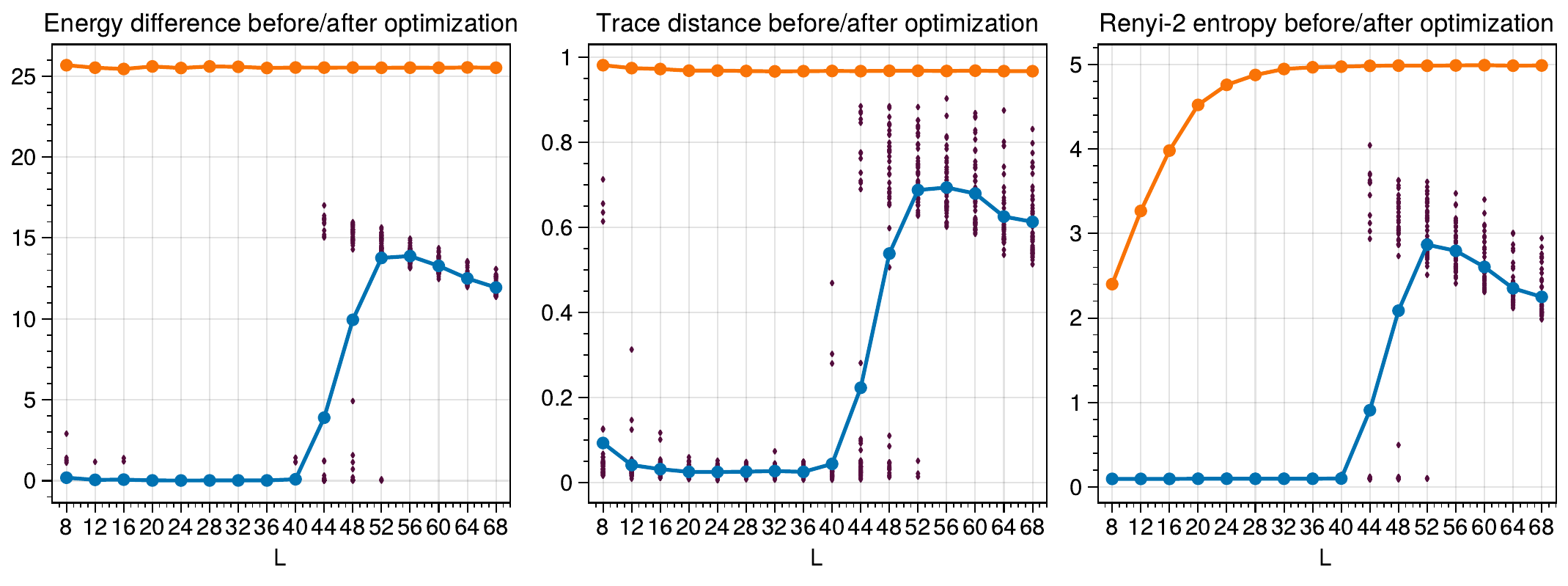}
    \caption{
    Measurements averaged over 50 independent circuits, before/after the VQA optimization with the nearest-neighbor Ising Hamiltonian \eqref{Ising1} at $g=2$, as a function of the number of circuit layers $L$.}
    \label{fig:IsingQ12Optimb}
\end{figure*}
\begin{figure*}[t]
    \centering
    \includegraphics[width=0.85\textwidth]{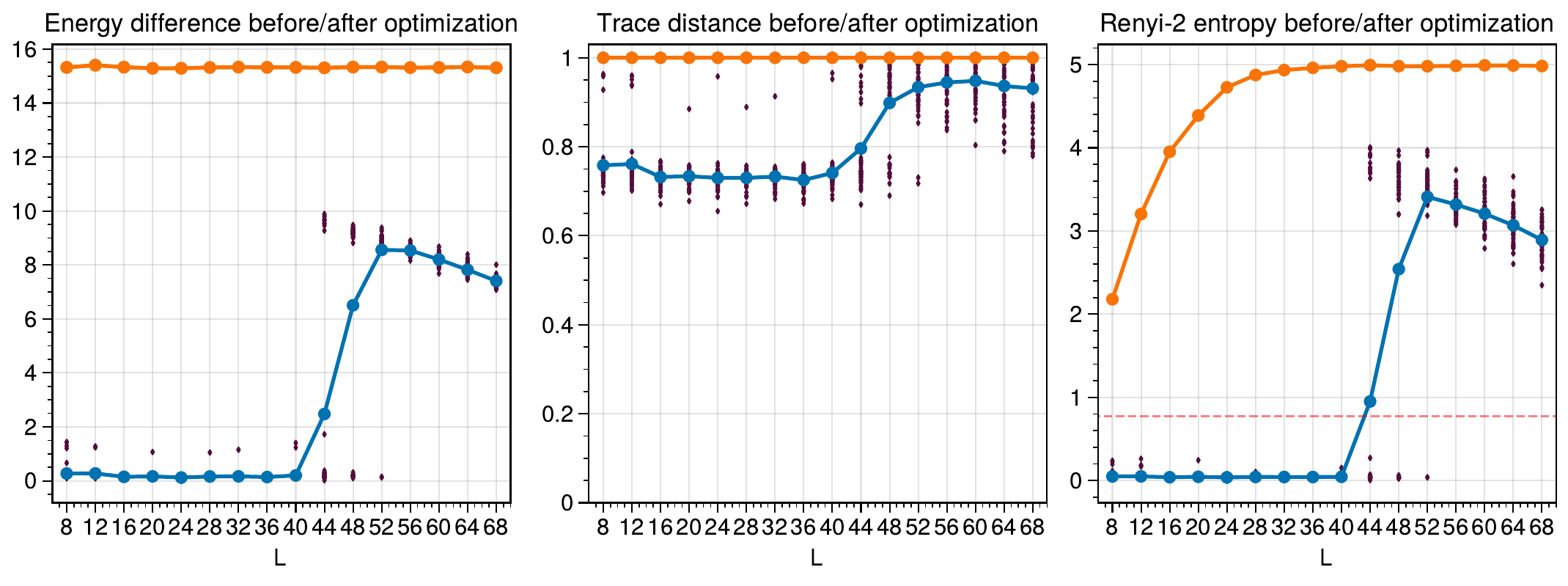}
    \caption{Measurements averaged over 50 independent circuits, before/after the VQA optimization with the nearest-neighbor Ising Hamiltonian \eqref{Ising1} at $g=1$, as a function of the number of circuit layers $L$.}
    \label{fig:IsingQ12G1Optim}
\end{figure*} 
The estimated values of $L_l$ and $L_s$ for different values of $n$ and $k$ are summarized in the fourth and sixth columns of Tables~\ref{tbl:random}~and~\ref{tbl:random-extra}. We make three observations: First, both timescales $L_l$ and $L_s$ increase as the entropy order $k$ goes higher. Second, the saturation time $L_s$ scales linearly in the system size $n$, i.e., $L_s \sim \mathcal{O}(n)$, because
\begin{equation}
 L_s  \gtrsim \frac{\text{Vol}(A) - c_k}{v_k \cdot \text{Area}(\partial A)} \sim \frac{n_A}{v_k} \sim \mathcal{O}(n).
\end{equation}
This is consistent with \cite{Liu_2018} that a unitary design that maximizes all Renyi entropies can be reached within a linear complexity in the system size $n$. Third, there exists a transient gap between $L_l$ and $L_s$, at least for finite-sized systems, in which the entanglement growth is slower. Details of the entanglement curves in this crossover region are largely model-dependent. See \cite{Liu_2014} for an example.

\begin{figure*}[t]
    \centering
    \includegraphics[width=0.85\textwidth]{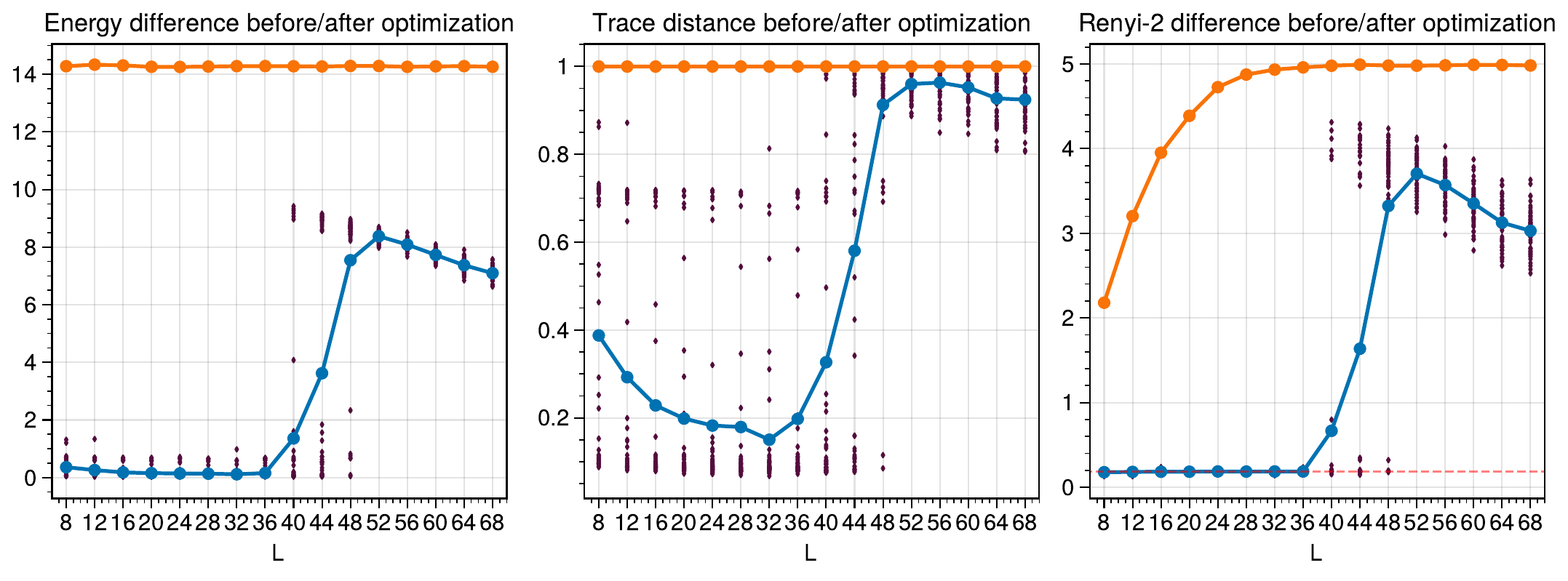}
    \caption{Measurements averaged over 50 independent circuits, before/after the VQA optimization with the non-local Ising Hamiltonian \eqref{Ising2} at $\alpha=g=1$, as a function of the number of circuit layers $L$.}
    \label{fig:IsingQ12Optim}
\end{figure*}
\begin{figure*}[t]
    \centering
    \includegraphics[width=0.85\textwidth]{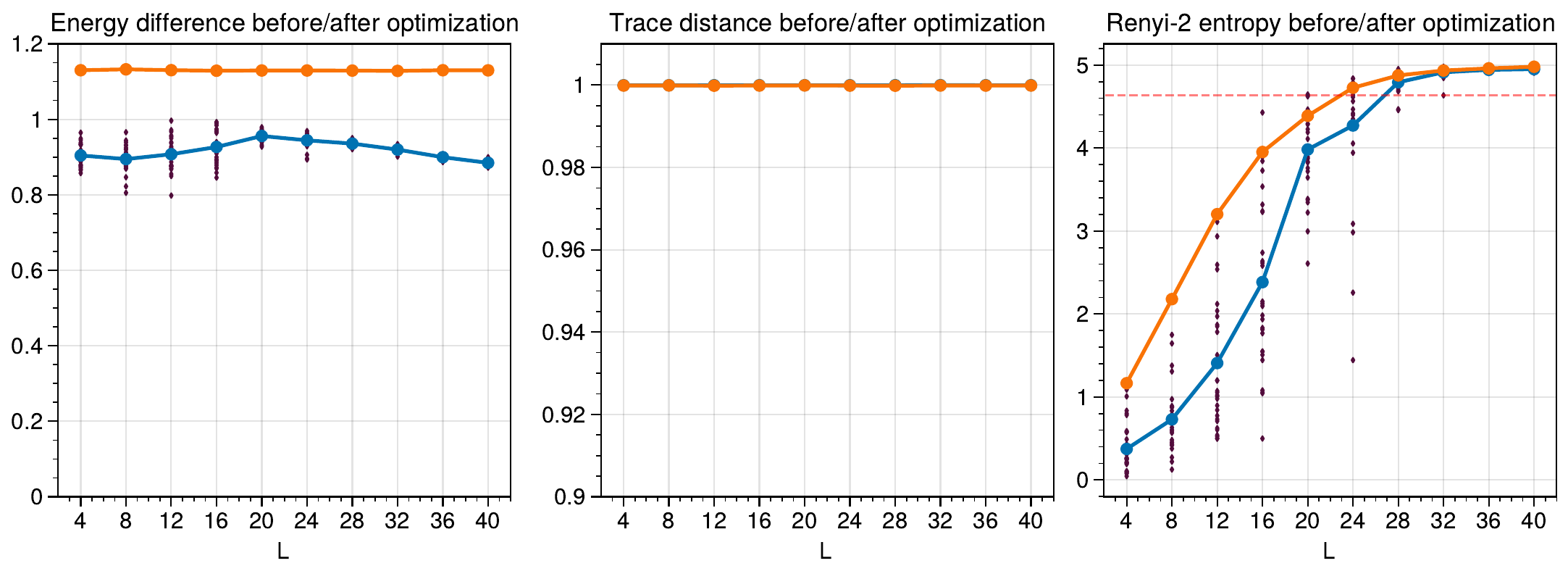}
    \caption{Measurements averaged over 30 independent circuits, before/after the VQA optimization with a particular instance of the SYK${}_{4}$ Hamiltonian \eqref{eq:syk_ham}, as a function of the number of circuit layers $L$.}    
    \label{fig:SYKQ12Optim}
\end{figure*}

\section{Circuit Optimization}
\label{sec:opt}

Our focus in this section is on the classical component of the hybrid quantum/classical algorithm. The objective is to find circuit parameters $\theta^*$ that closely approximate the ground state energy, $E(\theta^*) \simeq E_g$, by taking iterative steps proportional to the negative gradient of the energy function \eqref{E1} at each point $\theta_\tau$, i.e.,
\footnote{Estimating the gradient requires the readout of the circuit state $\rho_{c}$ at shifted gate parameters \cite{parameter_shift1} conducted by repeated measurements of Pauli strings. We will not consider the effect of the readout noise in this paper.}
\begin{align}
    \theta_{\tau + 1} = \theta_{\tau} - \eta \, \nabla E(\theta_{\tau}).
    \label{eq:gd}
\end{align}
The learning rate $\eta$ 
scales the step size 
of each update. 
A too-large $\eta$ 
can cause overshooting near the minimum $\theta^*$,
while a too-small $\eta$ 
can make the optimization trajectory
stuck at local minima. 
We will use
$\eta=0.005$ for most experiments. 
When the parameters update is small, each step of the gradient descent
can reduce the energy by 
\begin{align}
    \Delta E(\theta_{\tau}) \equiv E(\theta_{\tau+1}) - E(\theta_{\tau}) = -\eta\, \norm{\nabla E(\theta_\tau)}_2^2 \ .
    \label{eq:energy-dec}
\end{align}
Due to the constant decrease of the energy \eqref{eq:energy-dec}, we expect to  reach $E_g$ eventually if there are no other obstacles. We will terminate the iteration after updating the circuit parameter $10^4$ times in all our numerical experiments.

\subsection{Results}
\label{sec:res}

Let us discuss the eigensolver optimization results that aim to solve the ground state of many-body systems. We specifically consider the 1d transverse-field Ising models with nearest-neighbor and long-range interactions and the Sachdev-Ye-Kitaev (SYK) models. See Appendix~\ref{sec:ham} for a brief review of their Hamiltonians and ground-state entanglement properties.

\subsubsection{The Transverse-Field Ising Models}
\label{sec:isings}
We search the ground states of interacting 1d spin-chain systems. To break the degeneracy of ground states, we turn on the magnetic coupling $g$ to all the spin variables, choosing it to be $g=1$ or $2$.
As we are interested in finding a general correlation between the entanglement diagnostics and the success of optimization, not relying on specific characteristics of Hamiltonians, we study the optimization for the following three Ising models: 
\begin{enumerate}[label=(\roman*),itemsep=-1ex,partopsep=1ex,parsep=1ex]
    \item the nearest-neighbor spin coupling \eqref{Ising1} with $g=2$
    \item the nearest-neighbor spin coupling \eqref{Ising1} with $g=1$
    \item the long-range spin coupling \eqref{Ising2} with $\alpha=g=1$.
\end{enumerate}
We repeatedly perform the circuit optimization 50 times, to remove fluctuation made by random parameter initialization, and record the circuit outputs in Figures~\ref{fig:IsingQ12Optimb}-\ref{fig:IsingQ12Optim} as a function of the circuit depth $L$. 

Each figure consists of three panels. The left ones represent the energy difference $E(\theta) - E_g$ between the circuit state \eqref{E1} and the exact ground state \eqref{eq:ising_eg}. The middle ones show the trace distance between the reduced circuit state $\rho_{c,A}$ and the reduced ground state $\rho_{g,A}$. The right ones display the Renyi-2 entropy of the reduced circuit state. All the orange/blue curves therein represent a corresponding quantity before/after the optimization.

\begin{figure*}[t]
    \centering
    \includegraphics[width=0.90\textwidth]{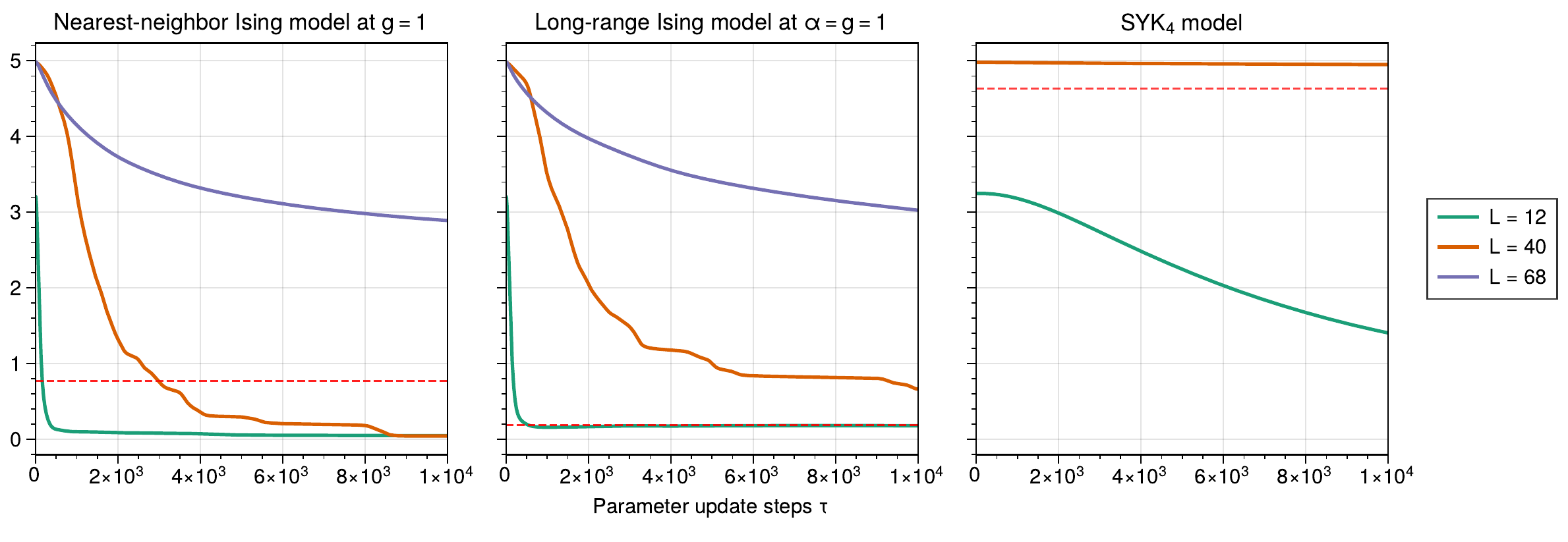}
    \caption{The evolution of the Renyi-$2$ entropy by the gradient-based VQA optimization with the Ising and SYK${}_4$ Hamiltonians, i.e., \eqref{Ising1}, \eqref{Ising2}, \eqref{eq:syk_ham}, for $n=12$ qubits. The horizontal direction denotes the number of parameter updates $\tau$. The dashed line shows the Renyi-$2$ entropy of the exact ground states for each Hamiltonian.}
    \label{fig:training-curve}
\end{figure*} 

Figure~\ref{fig:IsingQ12Optimb} is for the nearest-neighbor Ising model \eqref{Ising1} with $g=2$. It reveals the relation between the average entanglement entropy of initial circuit states and the success of gradient-based optimization: The optimization works well for the circuits with $L \lesssim 36$. However, in the intermediate range of $36 \lesssim L \lesssim 52$, the success rate gradually lowers as the circuit becomes deeper.  Beyond that, i.e., $L \gtrsim 52$, it always fails to close the gap between the exact ground state and the circuit state as to their energy and entanglement entropies. Such relation shows an advantage of using the circuits with $L < L_s$, whose entanglement curve has not reached the plateau. The above pattern 
also persists in Figures~\ref{fig:IsingQ12G1Optim}~and~\ref{fig:IsingQ12Optim}, which correspond to the nearest-neighbor and long-range Ising models with $g=1$.

 A notable difference of Figure~\ref{fig:IsingQ12G1Optim} from Figures~\ref{fig:IsingQ12Optimb} or~\ref{fig:IsingQ12Optim} appears in the trace distance curve, where the optimization fails to narrow the distance even when the circuit energy is close to the exact ground state energy. It is related to the fact that the ground state entanglement entropy in the $g=1$ nearest-neighbor Ising model is higher than those in other Ising models, as shown in Figure~\ref{re}. When the entanglement entropy of target ground states is higher, the local search of approximating circuit parameters becomes increasingly difficult. Such difficulty leads to deviations between the post-optimization circuit state and the exact ground state, to which the trace distance reflects much more sensitively than the energy and entanglement entropy differences.

\subsubsection{The Sachdev-Ye-Kitaev Model}
\label{sec:sykoptim}

 We will now discuss the circuit optimization in a situation where the Hamiltonian ground state exhibits
a volume law scaling of the entanglement entropy. 
 
The SYK${}_4$ Hamiltonian \eqref{eq:syk_ham} defined with an instance of random coupling constants has a ground state that follows the volume law scaling of entanglement \cite{Huang_2019}, as reviewed in Appendix~\ref{sec:ham} and specifically in Figure~\ref{re}. We optimize the variational circuit to approximate the SYK${}_4$ ground state and summarize the output in Figure~\ref{fig:SYKQ12Optim} as a function of the circuit depth $L$.

Since the approximation target state itself behaves in terms of entanglement like a generic quantum state, the optimization task is now much more challenging. Unlike the optimization towards the Ising ground state, even choosing a less entangled circuit within the range $L \lesssim 36$ does not lead to success. Figure~\ref{fig:SYKQ12Optim} manifests this failure, not only in the trace distance between the circuit and exact ground states but also their differences of energy and entanglement entropy.

\subsubsection{Optimization Speed}

As another indicator of how difficult the circuit optimization is, we draw in Figure~\ref{fig:training-curve} the evolution of Renyi-2 entanglement entropy $\mathcal{R}^2_A(\rho_{c,A})$ as a function of the number of parameter updates $\tau$. The three curves therein are for the circuits with $L=12$, $40$, $68$, which represent the characteristics of low-, intermediate-, high-depth circuits. The entanglement entropy of the ground states is marked by the dashed lines.

Towards the Ising ground states, the gradient descent works efficiently for the $L=12$ circuit, rapidly reducing the entanglement entropy  within $10^3$ steps of the update. However, the same gradient descent takes a much longer time for the $L=40$ circuit and even fails to reach closely the target state for the $L=68$ circuit. As the average entanglement entropy of initial circuit states increases, the optimization difficulty becomes more evident not only in the trace distance, as in Figure~\ref{fig:IsingQ12G1Optim}, but also through the entanglement diagnostics.

The optimization task towards the SYK${}_4$ ground state is inherently more challenging such that all three curves leave a large entanglement gap from the target state. Interestingly, the gradient descent constantly reduces the entanglement entropy of the $L=12$ circuit state, enlarging the gap over the optimization steps $\tau$. 
In general, over-parameterization can assist the circuit optimization that starts from/ends at a highly-entangled typical quantum state. An exponentially high-dimensional parameter space
was needed for the SYK${}_4$ model to approximate its ground-level energy with very high precision \cite{highdepth}.

\subsection{Entanglement Diagnostics and Optimization}

Our results shown in Section~\ref{sec:res} exemplify the difficulty in finding a successful optimization trajectory that starts from or ends at a typical quantum state that takes up the vast majority of the Hilbert space. This has been best described through the evolution of the entanglement entropies, \eqref{Re} and \eqref{SEE}, over the optimization steps, rather than a more commonly-used sensitive metric such as the trace distance between the circuit and target states.

Suppose we can divide the Hilbert space into two subregions distinguished by their entanglement entropies, say $A$ and $B/C$, in accordance with Figure~\ref{fig:random_growth}. Generic random states belong to the region $B/C$ whose entanglement entropies are approximately maximal.

For many interesting applications, the target state is a non-generic state that resides in the region $A$, i.e., following the area-law scaling of the entanglement \cite{Lucas_2014}. Along an optimization path inside the region $A$,  the  circuit state entanglement tends to decrease regularly. However, when an initial state $\rho_{c}(\theta_\text{in})$ belongs to the region $B/C$, the local parameter update \eqref{eq:gd} is unable to cross over to the region $A$, thus failing to reach the ground state energy. We make these observations from the optimization result in Section~\ref{sec:isings} that discusses the Ising Hamiltonians.

\begin{figure}[tb]
\centering
\subfloat[Von Neumann Entropy]{
    \includegraphics[height=4.1cm]{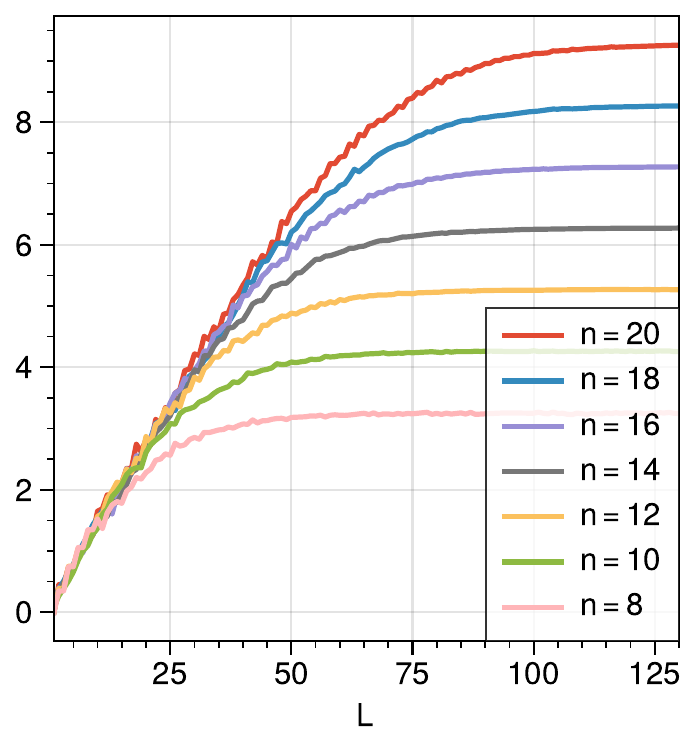}}
\subfloat[Renyi-2 Entropy]{
    \includegraphics[height=4.1cm]{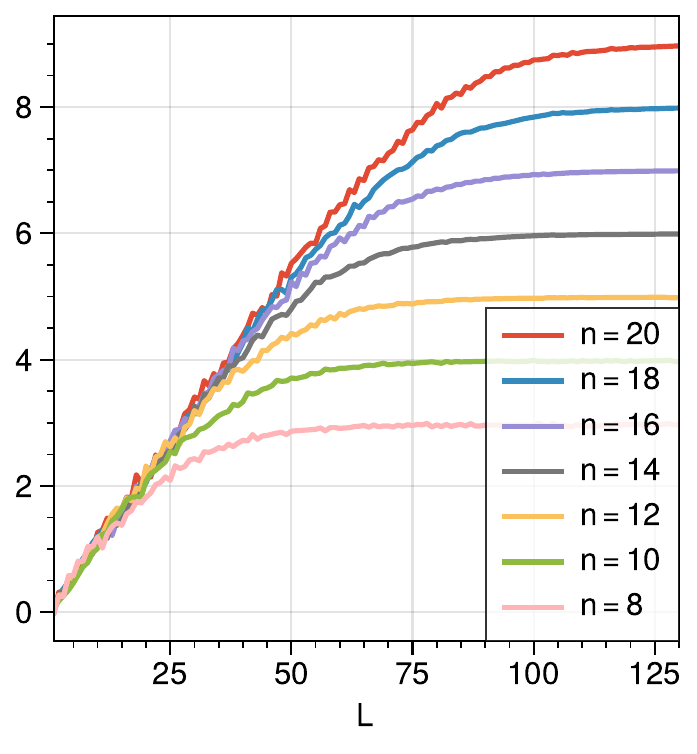}}
\caption{Von Neumann and Renyi-2 entropies for $8 \leq n \leq 20$ averaged over 50 independent $p=1/2$ circuit states, as a function of the circuit depth $L$.}
\label{fig:entropy-growth-stochastic}
\end{figure}

Even when the target state is maximally entangled and lies in the region $B/C$, the entanglement entropy of the circuit state $\rho_{c}(\theta_{\tau})$ still tends to decrease on average. It means that the entanglement gap between the circuit and target state can become larger, if an initial circuit state has a smaller entanglement entropy than the target state, i.e., $\mathcal{R}^{k}_A(\rho_{c, A}) \lesssim \mathcal{R}^{k}_A(\rho_{g, A})$.
When $\mathcal{R}^{k}_A(\rho_{c, A}) > \mathcal{R}^{k}_A(\rho_{g, A})$, the optimization moves towards narrowing the gap, but often failing to match a desired level of the entanglement. These observations are based on the optimization results in Section~\ref{sec:sykoptim}, obtained for the 1d SYK Hamiltonian.

The numerical results suggest that the Hilbert space can be partitioned into multiple layers, distinguished by the supported amount of the bipartite entanglement entropy. It is a very demanding task to move across distant layers via the gradient descent \eqref{eq:gd}, which is doable only for the over-parameterized circuits that involve exponentially large parameter space \cite{highdepth}.

\begin{figure}[tb]
\centering
\subfloat[$n=12$]{
    \includegraphics[height=4.1cm]{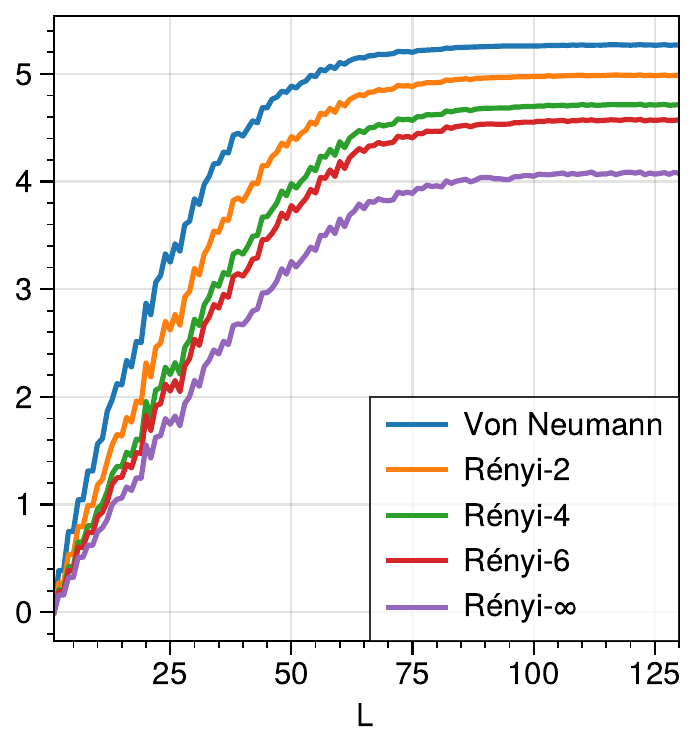}}
\subfloat[$n=20$]{
    \includegraphics[height=4.1cm]{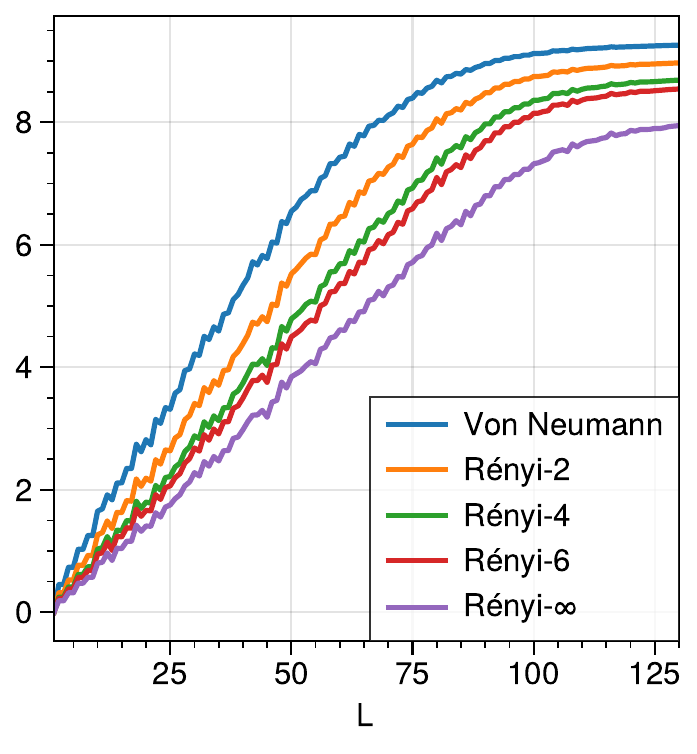}}
\caption{Various Renyi-$k$ entropies for $n=12$ and $n=20$ averaged over 50 independent  $p=1/2$ circuit states, as a function of the circuit depth $L$.}
\label{fig:entropies-generic1}
\end{figure}

\section{Other Circuit Architectures}
\label{sec:other}
We discussed the importance of choosing the circuit to avoid the saturation of its average initial entanglement entropy, for a generic optimization task that finds a target state with the area law entanglement. 
This section examines if the entanglement diagnostic still serves as an indicator of efficient circuit optimization with different circuit architectures. We also consider the effect of reducing the number of circuit parameters while maintaining a similar degree of entanglement.

\begin{figure*}[!t]
    \centering
    \includegraphics[width=0.85\textwidth]{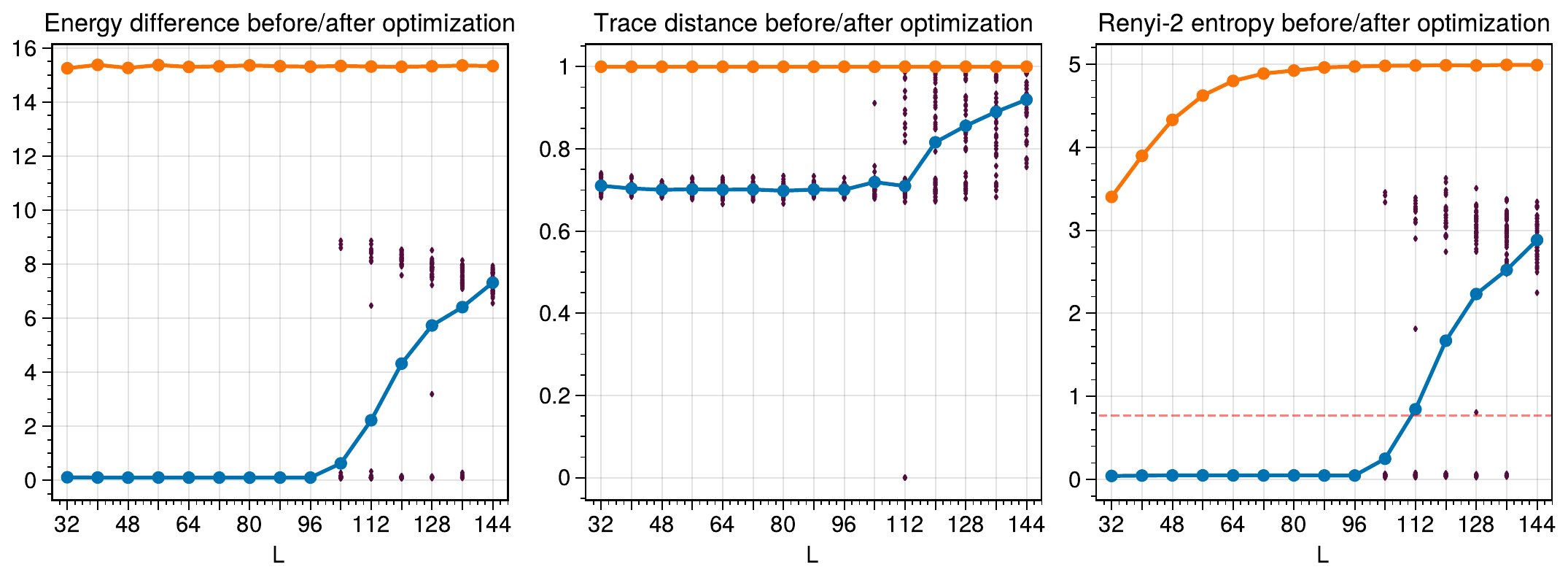}
    \caption{Measurements averaged over 50 independent $p=1/2$ circuits, before/after the VQA optimization with the nearest-neighbor Ising Hamiltonian \eqref{Ising1} at $g=1$, as a function of the number of circuit layers $L$.}    
    \label{fig:IsingQ12Optim-stochastic}
\end{figure*}
\begin{figure*}[t!]
    \centering
    \includegraphics[width=0.85\textwidth]{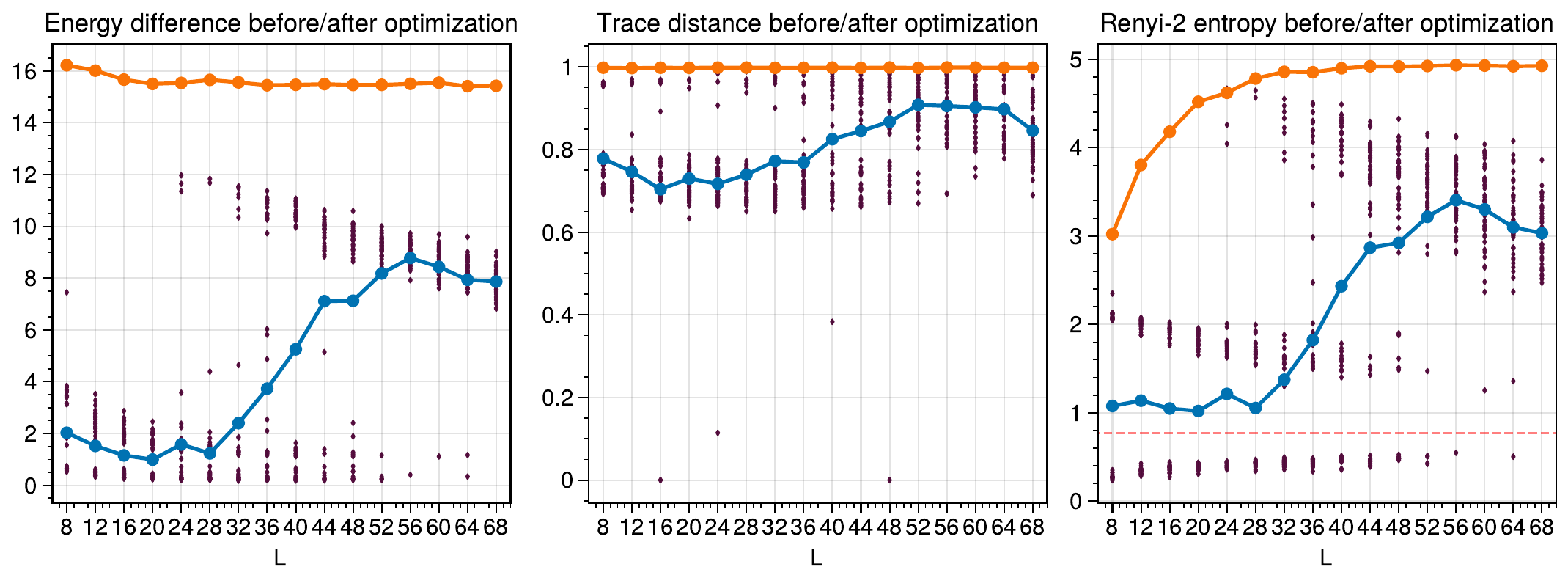}
    \caption{Measurements averaged over 50 independent restricted circuits \eqref{sub}, before/after the VQA optimization with the nearest-neighbor Ising Hamiltonian \eqref{Ising1} at $g=1$, as a function of the number of circuit layers $L$.}    
    \label{fig:IsingQ12Optim-layerequal}
\end{figure*}

\subsection{Random Graph Architecture}
\label{sec:random}

Let us study a simple stochastic variation of the circuit architecture that omits the $CZ$ entangler \eqref{eq:cz} inside the 2-qubit gate of Figure~\ref{fig:circuit_diagram_right} with a fixed probability $p=1/2$.

\subsubsection{Entanglement Growth}
Since the average number of the entangler is cut in half, we expect that the entanglement growth rate would be halved. Accordingly, the circuit depth $L$ to reach the saturation of the entanglement entropy would be doubled.

Figures~\ref{fig:entropy-growth-stochastic}~and~\ref{fig:entropies-generic1} show the evolution of entanglement diagnostics as a function of the circuit depth $L$, estimated by the sample averages of 50 random states.
 The overall shape of the curves remains the same, but the growth rate has significantly decreased. Reaching a certain level of the entanglement diagnostics requires twice the circuit depth compared to the non-stochastic  architecture, i.e., $p=1$, as expected.
 See Figure~\ref{fig:geometry-prob} for the curve of the geometric measure whose growth rate has been halved.

\subsubsection{Optimization}

Given the optimization task that reaches the nearest-neighbor Ising ground energy \eqref{eq:ising_eg} with the background field coupling $g=1$, the outputs of the $p=1/2$ stochastic circuit are all collected in  Figure~\ref{fig:IsingQ12Optim-stochastic} as a function of $L$.

The depth range of the $p=1/2$ circuits where the gradient descent remains successful has increased to $L \lesssim 96$. Beyond that, the optimization success rate continues to drop until it reaches $0\%$ at $L \sim 144$ and above. This is consistent with the entanglement growth curves, which  maintain the same overall shape as in Section~\ref{sec:randomcircuit} but only with a lower growth rate. We remark that the low and intermediate ranges, in which the optimization may succeed with a non-zero probability, has been extended to $L \lesssim 136$, more than mere doubling. It is the impact of the expanded parameter space whose dimension has been doubled, as required for the $p=1/2$ circuit to hold the same level of entanglement.

Over the entire range of $L$, unlike the trace distance, the entropy diagnostic holds a robust correlation with the successful minimization of the circuit energy \eqref{E1}, showing its usefulness regardless of circuit-specific details.

\subsection{Restricted Circuit Parametrization}
\label{sec:param}
Recall that an additional circuit layer  increases both the average entanglement entropy of initial circuit states and the number of classical parameters. To isolate the effect of the classical parameter space, we study the consequence of imposing the following restriction:
\begin{equation}
\theta_{\ell,1} = \theta_{\ell, 2} =  \cdots = \theta_{\ell, n} \quad \text{ for all } 1\leq \ell \leq L \ ,
\label{sub}
\end{equation}
which equates all the parameters in each layer, 
yet maintains the same growth rate of  entanglement diagnostics.

\begin{figure}[!t]
\centering
\subfloat[Von Neumann Entropy]{
    \includegraphics[height=4.1cm]{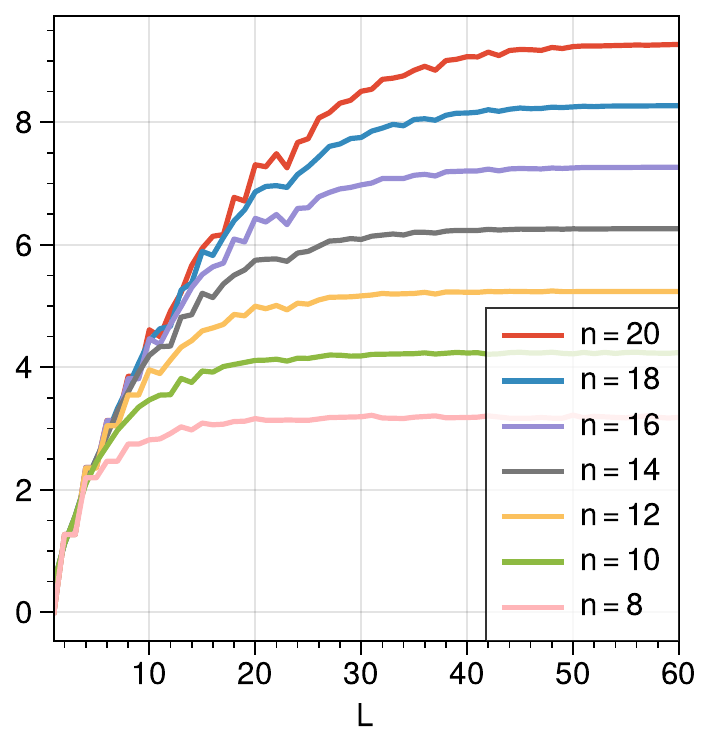}}
\subfloat[Renyi-2 Entropy]{
    \includegraphics[height=4.1cm]{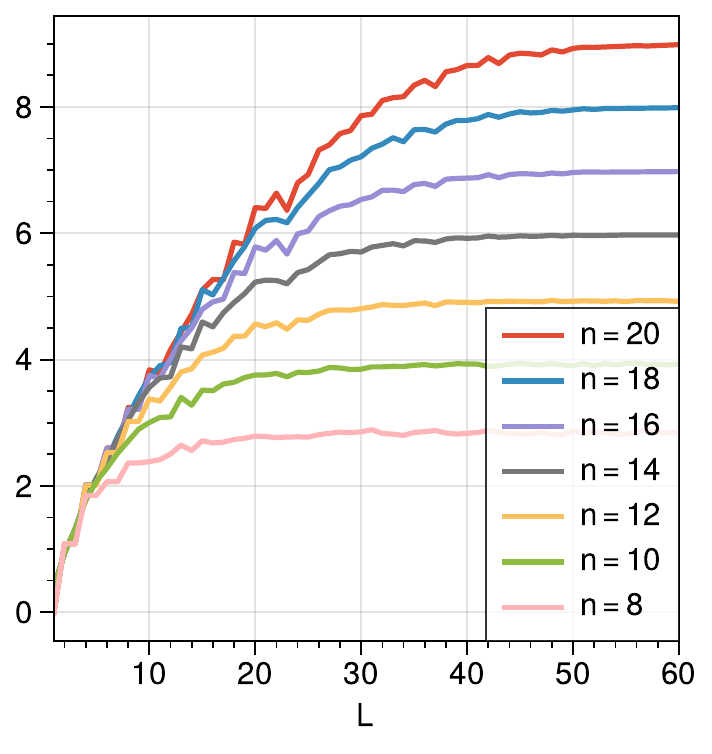}}
\caption{Von Neumann and Renyi-2 entropies for $8 \leq n \leq 20$ averaged over 50 independent restricted circuit states \eqref{sub}, as a function of the circuit depth $L$.}
\label{fig:entropy-growth-layerequal}
\end{figure}

The basic 2-qubit gate ${\cal O}_{i,j}$ 
in Figure~\ref{fig:circuit_diagram_right} reads:
\begin{equation}
{\cal O}_{i,j} = CZ_{i,j}  \cdot R(\theta_{l,i})\otimes R(\theta_{l,j}) \ ,    
\end{equation}
where
$CZ=\text{diag}(1,1,1-1)$ and $R(\theta)$ is the Pauli rotation \eqref{eq:paulirot} around the $y$-axis.
It is curious to note that the constraint
(\ref{sub}) is equivalent
to imposing $[{\cal O}_{i,j},Q_{i,j}] = 0$ 
on the Hilbert space of $(i,j)$ qubits, where:
\begin{equation}
Q_{i,j} =
\begin{pmatrix}
q_1+q_2  & 0 & 0 & 0\\
0 & q_1 & q_2  & 0 \\
  0 & q_2 & q_1  & 0 \\ 
   0 & 0 & 0  &  q_1+q_2 \\
 \end{pmatrix}    
 \end{equation}
in the computational basis of $(i,j)$ qubits.
Still, there is no globally conserved charge written as a tensor product sum of $Q_{i,i+1}$,
because it does not generically commute with ${\cal O}_{i-1,i}$ and ${\cal O}_{i,i+1}$ on the next layer.

\subsubsection{Entanglement Growth}

The entanglement entropies averaged over 50 random circuit states under the parameter space restriction \eqref{sub}  are illustrated in Figures~\ref{fig:entropy-growth-layerequal}~and~\ref{fig:entropies-layerequal} as a function of the number of circuit layers $L$. Except small extra wiggles, the overall growth shape and speed of the entanglement diagnostics are similar to those of the unconstrained circuit. 
Such correspondence of the entanglement growth curves renders the restricted circuit  an appropriate setup to study separately the effect of the parameter space dimension on the circuit optimization.

As a side remark, we have seen that the evolution curve of the geometric measure, illustrated in Figure~\ref{fig:geometry-rest}, for the restricted circuit is far more fluctuating than as for the unconstrained circuit, while their saturation depth scales remain largely the same.

\begin{figure}[t]
\centering
\subfloat[$n=12$]{
    \includegraphics[height=4.1cm]{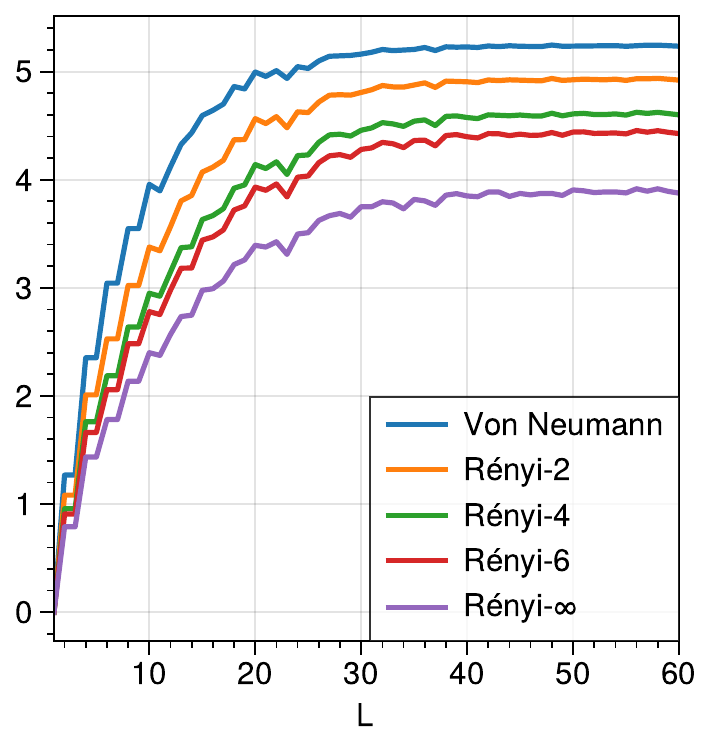}}
\subfloat[$n=20$]{
    \includegraphics[height=4.1cm]{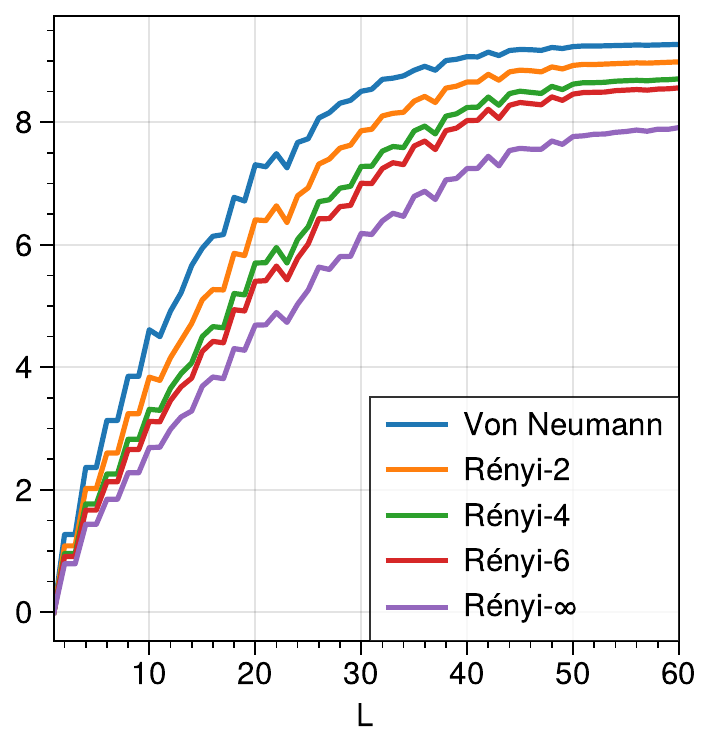}}
\caption{Various Renyi-$k$ entropies for $n=12$ and $n=20$ averaged over 50 independent restricted circuit states \eqref{sub}, as a function of the circuit depth $L$.}
\label{fig:entropies-layerequal}
\end{figure}

\subsubsection{Optimization}

We optimize the restricted circuit to approximate the ground state of the nearest-neighbor Ising model at $g=1$ using the gradient descent. The results are summarized in Figure~\ref{fig:IsingQ12Optim-layerequal} as a function of $L$. It is notable that even the circuit optimization with $L \lesssim 20$ often stops at $\Delta E \gtrsim 1$, not giving a reliable approximation of the ground energy. Furthermore, starting from $L\gtrsim 24$, an increasingly large proportion of the randomly initialized circuits fails to reach the ground level energy $E_g$ and leave a large gap, i.e., $\Delta E \gtrsim 8$. Such transitional result emerges at a much lower depth than $L=44$ of the unconstrained circuit.

It emphasizes the importance of having enough parameters in applying the gradient descent to optimization tasks, even for those circuits that remain within a suitable range of the entanglement diagnostics.

\section{Discussion and Outlook}
\label{sec:dis}

In this paper we considered the variational circuit model of quantum computation, arguing for the effectiveness of entanglement diagnostics in finding the circuit architecture for efficient parameter optimization that minimizes the Hamiltonian expectation value. Introduced as a distance measure between the circuit and target states, the entanglement diagnostic has shown its usefulness by illustrating that quantum circuits states within a suitable range of entanglement entropies can successfully reach the ground level energy of local Hamiltonians. It also says that, while entanglement is a valuable non-local resource for quantum computation, circuit states being highly entangled do not necessarily have an advantage but it can be rather the opposite.\footnote{See also the discussion in \cite{GM}.}

One way to control the average entanglement entropy of circuit states is to adjust the number of circuit layers. The mean entropy grows linearly with the circuit depth, then gradually slows down, and finally converges to a constant near the theoretical maximum. Denoting by $L_s$ the saturation depth beyond which the average entanglement entropy has converged, we divided the depth range into two intervals, $L < L_s$ and $L \geq L_s$, and  called them respectively $A$ and $B/C$. $A$ is typically the optimal region that leads to efficient VQA computation, e.g., when we search an area-law entangled target state, while $B/C$ can suffer from the barren plateau problem, 
One can further differentiate $C$ from $B$ based on whether the optimization success rate has become $0\%$ or not yet.

Although the assumption of an area-law entangled target state covers most of the interesting VQA applications \cite{Lucas_2014}, the ground states of some important Hamiltonians exhibit volume-law entanglement scaling. Matching the entanglement diagnostic alone is not sufficient to approximate such states due to the overwhelming population of highly-entangled quantum states. We need deep variational circuits whose depth $L$ lies in $B/C$ and that are equipped with a large parameter space that can assist high-resolution specification and approximation of the desired target state \cite{highdepth}. Having more circuit parameters can generally help to approximate the ground state better, as exemplified by the decreased accuracy for a reduced number of independent variables\footnote{It is not a conflict with \cite{Funcke_2021} which reduces parameter space redundancy by identifying the principal directions carefully for maximal expressibility. Our reduction is rather arbitrary and without a guarantee that the circuit with remaining independent parameters can be maximally expressible.} in Section~\ref{sec:param}, as well as the increased success rate for circuits with extra single-qubit rotation parameters \cite{highdepth, circuit_hessian}.

There are many follow-up directions for further investigation: First, for having additional substantial evidence to the validity of the entanglement diagnostics, it would be crucial to consider 2d gapped local Hamiltonians whose ground states follow the area law entanglement scaling but are difficult to approximate. Second, we would like to explore various circuit architectures, e.g., using other rotation and entangling gates \cite{haug2021capacity}, built on different graph structures \cite{Harrow:2019lyw,Cervera_Lierta_2019}, or conserving diffusive charges \cite{sub,diffusive}.
Especially, symmetry-preserving circuits can work efficiently for the VQA optimization if the target state is known to respect the imposed built-in symmetries \cite{ising_circuit,ising_circuit2}. Third, the layered circuit defines a discrete dynamical system. We would like to investigate the appearance of quantum chaos in the circuit wavefunction, such as the emergence of random matrix ensemble for the reduced density matrix \cite{Chen_2018} and the operator spreading \cite{Mezei:2016wfz}, relating them to the efficient VQA optimization \cite{KO2}. Finally, it is important to study different types of noise and analyse how they
affect the VQA optimization performance \cite{barron2020measurement,fontana2020optimizing}.

\section*{Acknowledgements}

We would like to thank Khen Cohen, Tom Levy, Eun Gook Moon, Muli Safra and Lior Wolf
for valuable discussions.
We are grateful to Jaedeok Kim and Dario Rosa for the collaboration at the early stage of this project. The work of J.K. is supported by the NSF grant PHY-1911298 and the Sivian fund.
The work of Y.O. is supported in part by  Israel Science Foundation Center
of Excellence,
the IBM Einstein Fellowship and John and Maureen Hendricks Charitable Foundation
at the Institute for Advanced Study in Princeton.
Our Python code for the numerical experiments is written in TensorFlow Quantum \cite{TFQ}. The experimental data are managed by using Weights \& Biases \cite{wandb}.

 \begin{figure*}[p]
    \includegraphics[height=5.4cm]{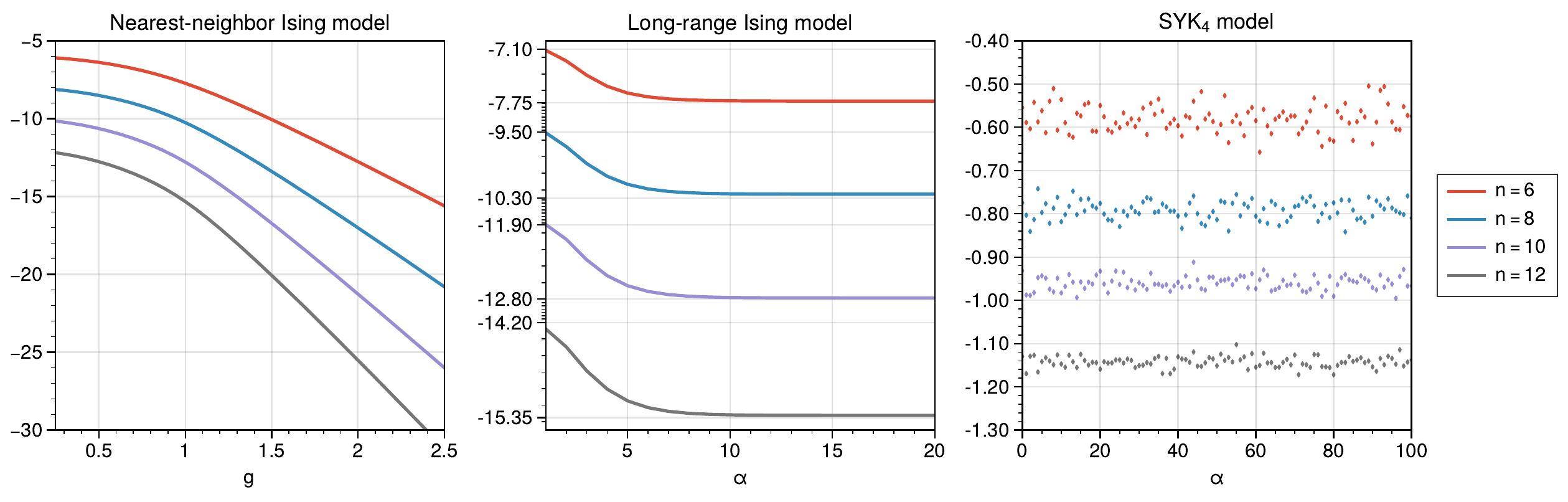}
    \caption{The ground energy of the following Hamiltonian systems of different sizes $n$: (Left) Nearest-neighbor Ising Hamiltonian with different $g$. (Middle) Long-range Ising Hamiltonian at $g=1$ with different $\alpha$. (Right) SYK${}_4$ Hamiltonian with 100 different instances of Gaussian random couplings.
    }
    \label{gse}
\end{figure*}
\appendix

\section{Hamiltonians}
\label{sec:ham}
Let us consider the following 1d Hamiltonian systems: the nearest-neighbor and long-range Ising models coupled to a transverse magnetic field and the Sachdev-Ye-Kitaev model \cite{SYK,2015Kitaev,Maldacena_2016}. Here we summarize some of their important characteristics, including the entanglement entropy scaling of their ground states.

\begin{figure*}[p]
\centering
    \includegraphics[height=5.6cm]{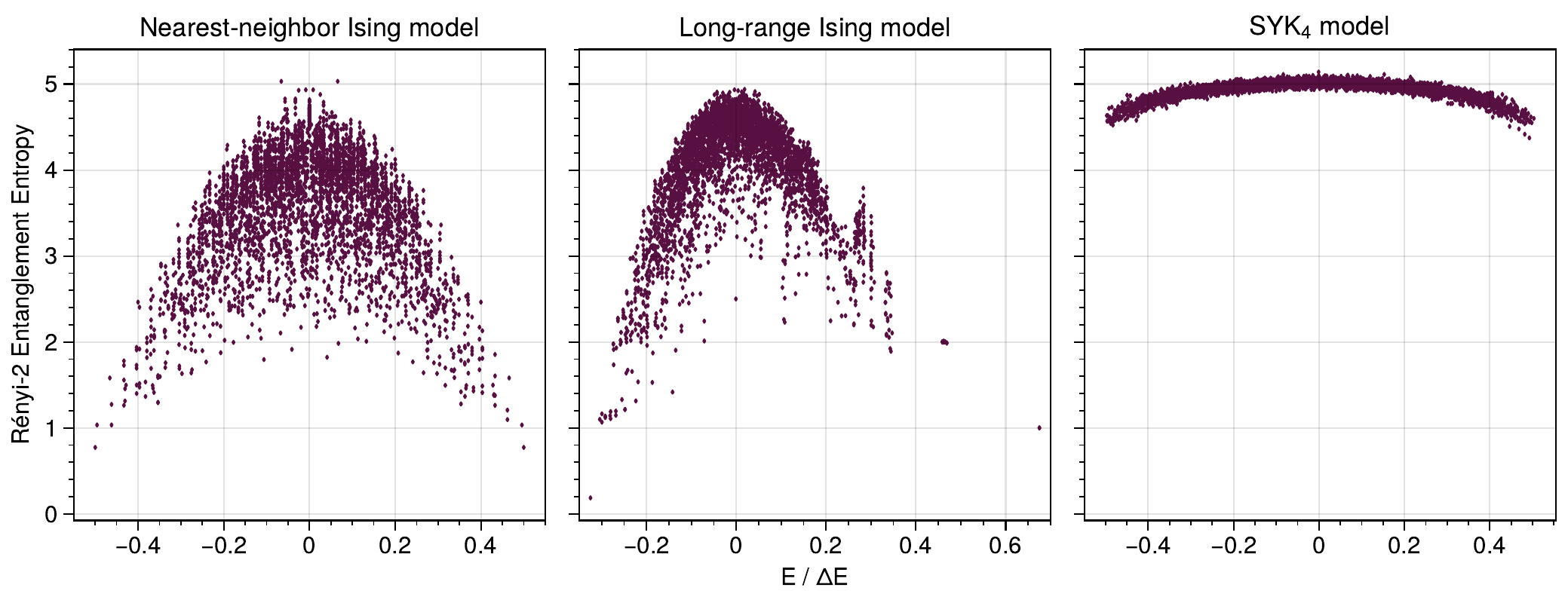}
\caption{The scatter plot of energy and Renyi-2 entropy over all eigenstates of the following $n=12$ Hamiltonians:  (Left) Nearest-neighbor Ising Hamiltonian with different $g$. (Middle) Long-range Ising Hamiltonian at $g=1$ with different $\alpha$. (Right) SYK${}_4$ Hamiltonian with 100 different instances of Gaussian random couplings.}
\label{re2}
\end{figure*}

\begin{figure*}[p]
\centering
    \includegraphics[height=5.7cm]{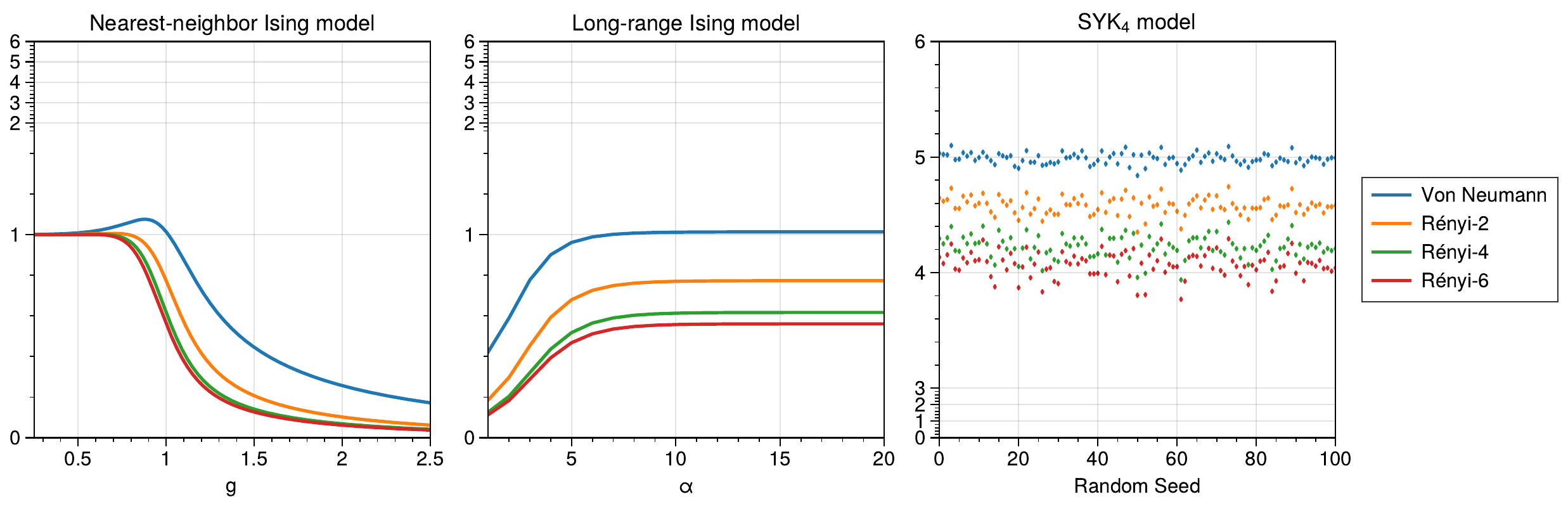}
    \caption{The ground state entanglement entropies of the following Hamiltonian systems of different sizes $n$:  (Left) Nearest-neighbor Ising Hamiltonian with different $g$. (Middle) Long-range Ising Hamiltonian at $g=1$ with different $\alpha$. (Right) SYK${}_4$ Hamiltonian with 100 different instances of Gaussian random couplings.}
\label{re}
\end{figure*}

\subsection{Nearest-Neighbor Ising Model}

The 1d Hamiltonian with the nearest-neighbor spin interaction coupled to a transverse magnetic field reads:
\begin{align}
    H =  \sum_{i=1}^n \sigma_{z,i} \sigma_{z,i+1} + g\sum_{i=1}^n \sigma_{x,i} \ ,
\label{Ising1}
\end{align}
where we assume the periodic boundary condition $i \sim i+n$ and the magnetic coupling $g$ being positive. 
$\sigma_{x/z,i}$ are Pauli-$x/z$ matrices acting on the $i$'th spin, respectively.

One can solve this Hamiltonian exactly, leading to the following ground-level energy \cite{ground}: 
\begin{equation}
E_{g} = - \sum_{k=0}^{n-1}\left(1 + g^2 - 2g \cos\left(\frac{2 \pi k}{n}\right)\right)^{1/2} \ .
\label{eq:ising_eg}
\end{equation}
The $g=0$ ground state is in the anti-ferromagnetic phase where all spin variables are aligned antiparallel to neighboring spins, such that $E_g = -n$. 
As $g$ grows, the spin-field coupling contributes more significantly to $E_g$.
Especially when $g \sim g_c$, the paramagnetic phase transition occurs such that all spins are now aligned in $-x$ direction. 
As $g \gg g_c$, the ground energy approaches to $E_g \simeq -gn$.
See Figure~\ref{gse} for the ground energy curves over the range of the magnetic coupling $0 \leq g \leq 2.5$ with different $n$. 

We also draw in Figure~\ref{re} the curves of various Renyi-$k$ entropies of the $n=12$ Ising ground state for $0 \leq g \leq 2.5$.
Figure~\ref{re2} contains the scatter plot of energy and Renyi-$2$ entanglement entropy over all the eigenstates of the $g=1$ Ising Hamiltonian with $n=12$ qubits.

\subsection{Long-Range Ising Model}
We add the long-range spin-spin interaction whose couplings decay with the distance. This Hamiltonian reads:
\begin{align}
H =   \sum_{i < j}\frac{1}{d(i,j)^\alpha}\,
\sigma_{z,i} \sigma_{z,j} + g\sum_{i=1}^n \sigma_{x,i}
\label{Ising2}
\end{align}
where $d(i,j)$ is the shortest distance between the $i$'th and $j$'th spins with the periodic boundary condition $i \sim i + n$. All the non-local interactions vanish in the limit $\alpha \rightarrow \infty$, thus the Hamiltonian \eqref{Ising2} reduces to \eqref{Ising1}.

As in the nearest-neighbor model, for any $\alpha \geq 0$, the long-range Ising model exhibits a transition between the anti-ferromagnetic and paramagnetic phases. 
Note that the ground state in the anti-ferromagnetic phase can have the entanglement entropy that grows with $n$, thus violating the area law.  Its scaling behavior in $n$ is logarithmic for $\alpha > 1$ and sub-logarithmic for $\alpha < 1$ \cite{longrange_ising}. Since the matrix product state ansatz can still closely approximate the ground state \cite{longrange_ising,mps_longrange_ising}, we expect the mild violation of the area law entanglement would not be a big obstacle of the gradient-based optimization even for large $n$. Several Renyi-$k$ entropies of the $n=12$ long-range Ising ground state are shown in Figure~\ref{re} as a function of  $0\leq \alpha \leq 20$.

Figure~\ref{gse} plots the $g=1$ ground energy as a function of the exponent $0\leq \alpha \leq 20$ for different system sizes $n$.
Since the long-range interactions are almost negligible for $\alpha \gtrsim 10$, the  curves converge to the energy \eqref{eq:ising_eg} of the nearest-neighbor Ising model  at $g=1$.  We also draw  the scatter plot of energy and Renyi-$2$ entanglement entropy in Figure~\ref{re2}, denoting every eigenstate of the $g=\alpha=1$ long-range Ising Hamiltonian with $n=12$ qubits.

\subsection{The SYK Model}

The Sachdev-Ye-Kitaev (SYK) model \cite{SYK,2015Kitaev} consists of random, long-range, all-to-all interactions of $n$ qubits, which correspond to the following random couplings of $q$ Majorana fermions:
\begin{equation}
H = (i)^{q/2} \sum_{1 \leq i_1<...<i_q \leq 2n} J_{i_1...i_q} \gamma_{i_1}...\gamma_{i_q} \ ,
\label{eq:syk_ham}
\end{equation}
where the Majorana fermions $\{\gamma_i\}_{1\leq i\leq 2n}$ satisfy the Clifford algebra  $\{\gamma_i,\gamma_j\} = \delta_{ij}$ and can be translated to the spin variables via the Jorgan-Wigner map. The coupling constants $J_{i_1...i_q}$ are randomly drawn from the Gaussian distribution with mean $0$ and variance ${(q-1)!}/{(2n)^{q-1}}$.
Much attention has been paid to the SYK${}_q$ model because it is exactly solvable and exhibits a chaotic dynamics for $q\geq 4$ at the same time \cite{2015Kitaev,Maldacena_2016}.

We focus on the SYK${}_4$ model and drop the subscript for brevity. 
Each random draw of the coupling constants $J_{i_1...i_4}$ from the Gaussian distribution defines a different instance of the SYK Hamiltonian. The SYK ground energy for 100 individual instances with $n$ qubits, or equivalently, $2n$ Majorana fermions are displayed in Figure~\ref{gse}.
Similarly, the entanglement entropies of 100  instances of the $n=12$ SYK ground state
 are visualized in Figure~\ref{re}, illustrating  the SYK ground state is much more highly-entangled than that of the Ising models. More generally, the  energy-entropy scatter plot of Figure~\ref{re2} denotes the full spectrum for an instance of the $n=12$ SYK Hamiltonian, exhibiting the  volume-law scaling of the entanglement entropies  \cite{Huang_2019}. Its energy gap is notably smaller than that of the Ising models, thus violating the assumption \cite{Hastings_2007} for the area-law entanglement of the ground state.

\begin{table}[b]
\centering
\begin{tabular}{@{}cccccccccccc@{}}
\toprule
Type &  $n$ & $v_k$ & $L_l$ & $r_{n,k}$ & $L_s$ & Type &  $n$ & $v_k$ & $L_l$ & $r_{n,k}$ & $L_s$\\ \midrule
\multirow{7}{*}{$S_\text{max}$} 
& $8$ & 0.8534 & 6 & 3.9985 & 9 & \multirow{7}{*}{$S_\text{min}$} & $8$ & 0.1724 & 9 & 2.2341 & 23\\
& $10$ & 0.9342 & 6 & 4.9983 & 12&& $10$ & 0.1611 & 14 & 3.1373 & 33\\
& $12$ & 0.8694 & 8 & 5.9974 & 18 && $12$ & 0.1651 & 21 & 4.0823 & 39\\
& $14$ & 0.8922 & 8 & 6.9965 & 23 && $14$ & 0.1739 & 23 & 5.0491 & 46\\
& $16$ & 0.8510 & 9 & 7.9956 & 32 && $16$ & 0.1715 & 35 & 6.0303 & 53\\
& $18$ & 0.8305 & 10 & 8.9947 & 35 && $18$ & 0.1726 & 39 & 7.0181 & 60\\
& $20$ & 0.8011 & 11 & 9.9931 & 47 && $20$ & 0.1686 & 49 & 8.0118 & 66\\\midrule
\multirow{7}{*}{$\mathcal{R}^4_A$} 
& $8$ & 0.2232 & 9 & 2.7084 & 29 & \multirow{7}{*}{$\mathcal{R}^6_A$} & $8$ & 0.2050 & 9 & 2.5798 & 27\\
& $10$ & 0.2077 & 12 & 3.7101 & 33 && $10$ & 0.1911 & 14 & 3.5708 & 33\\
& $12$ & 0.2134 & 19 & 4.7177 & 43  && $12$ & 0.1963 & 21 & 4.5764 & 41\\
& $14$ & 0.2225 & 19 & 5.7234 & 49 && $14$ & 0.2057 & 22 & 5.5823 & 49\\
& $16$ & 0.2199 & 29 & 6.7272 & 60 && $16$ & 0.2031 & 31 & 6.5866 & 60\\
& $18$ & 0.2212 & 32 & 7.7290 & 63 && $18$ & 0.2046 & 35 & 7.5887 & 63\\
& $20$ & 0.2156 & 41 & 8.7299 & 71 && $20$ & 0.1995 & 43 & 8.5899 & 71\\\bottomrule

\end{tabular}
\caption{Estimation of the entanglement velocity, late-time saturation value, early-time scale, and the late-time scale.}
\label{tbl:random-extra}
\end{table}

\begin{figure*}[h]
\centering
\subfloat[The circuit of Figure~\ref{fig:circuit_diagram}]{
    \includegraphics[height=5cm]{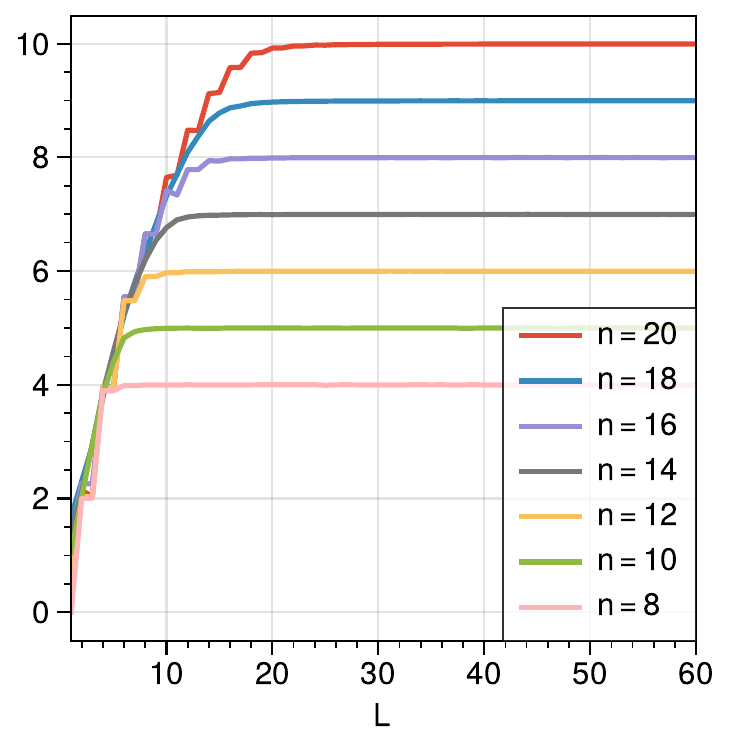}}
\subfloat[The $p=1/2$ circuit]{
    \includegraphics[height=5cm]{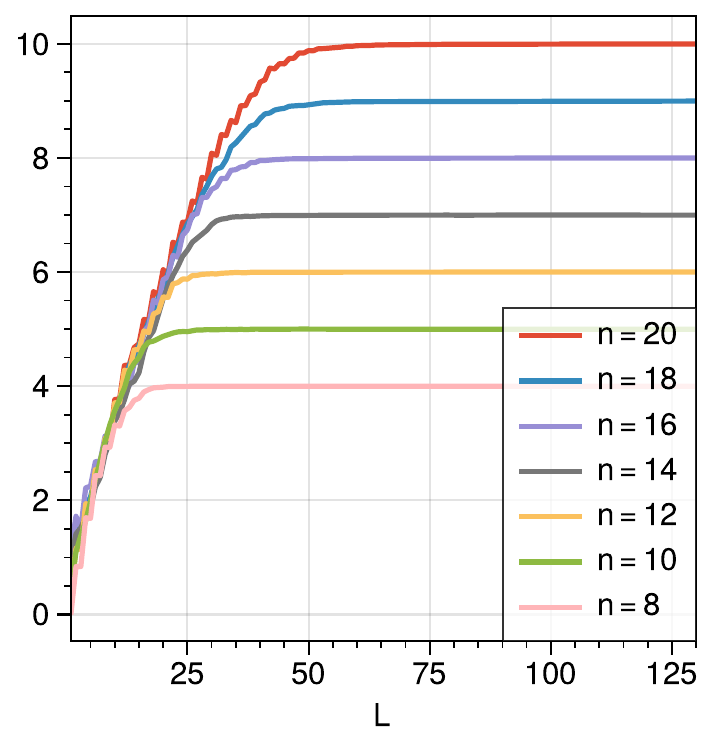}}
\subfloat[The restricted circuit]{
    \includegraphics[height=5cm]{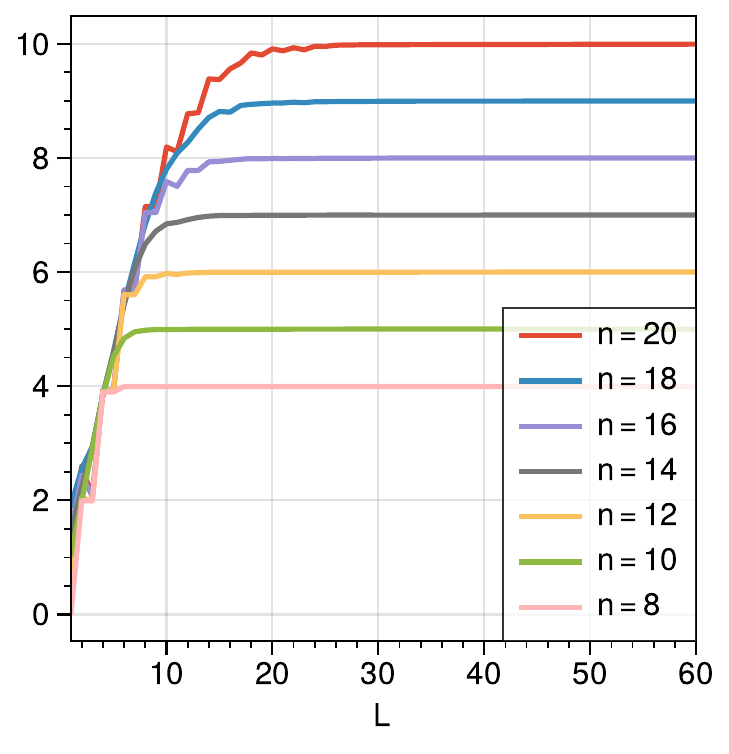}}
\caption{The max-entropy of the circuit reduced density matrix $\rho_A(\theta)$ as a function of the number of layers $L$.}
\label{fig:pmax}
\end{figure*}
\begin{figure*}[h]
\centering
\subfloat[The circuit of Figure~\ref{fig:circuit_diagram}]{
    \includegraphics[height=5cm]{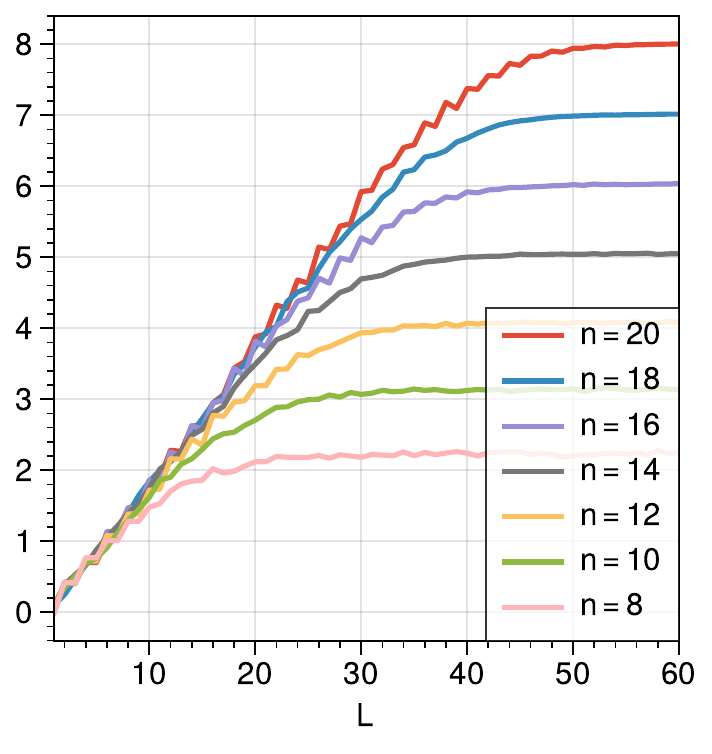}}
\subfloat[The $p=1/2$ circuit]{
    \includegraphics[height=5cm]{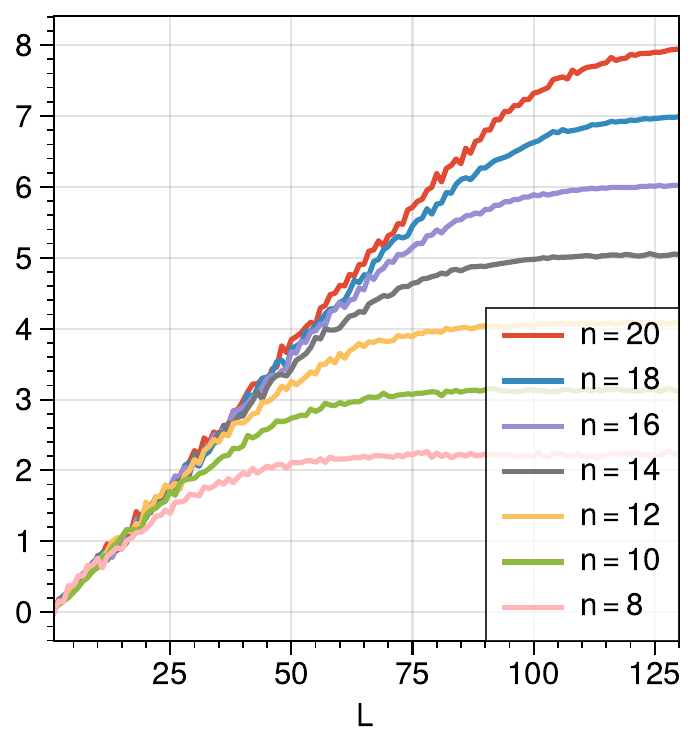}}
\subfloat[The restricted circuit]{
    \includegraphics[height=5cm]{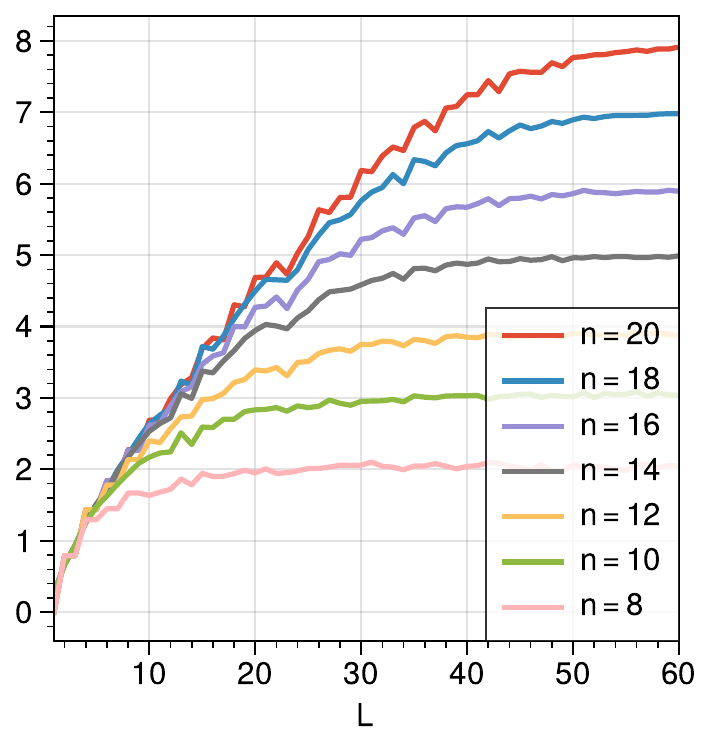}}
\caption{The min-entropy of the circuit reduced density matrix $\rho_A(\theta)$ as a function of the number of layers $L$.}
    \label{fig:pmin}
\end{figure*}
\begin{figure*}[h]
    \centering
\subfloat[The circuit of Figure~\ref{fig:circuit_diagram}]{
    \includegraphics[height=5cm]{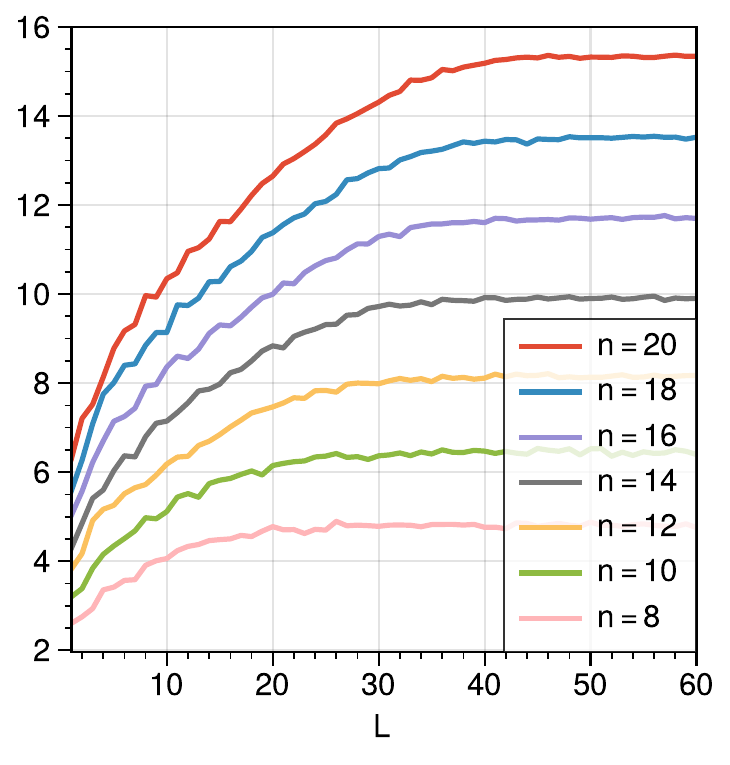}}
\subfloat[The $p=1/2$ circuit\label{fig:geometry-prob}]{
    \includegraphics[height=5cm]{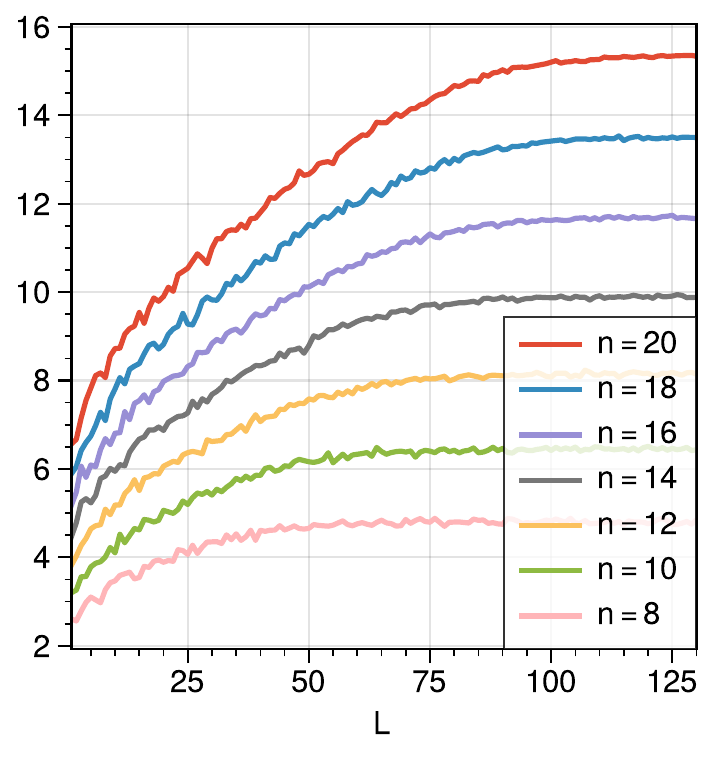}}
\subfloat[The restricted circuit\label{fig:geometry-rest}]{
    \includegraphics[height=5cm]{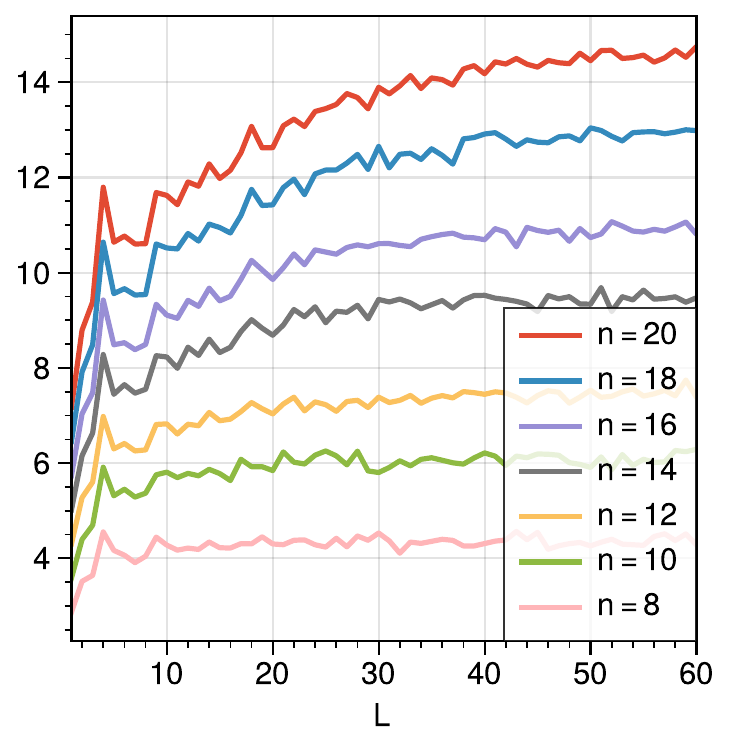}}
\caption{The geometric measure of the circuit reduced density matrix $\rho_A(\theta)$ as a function of the number of layers $L$.}
\label{fig:geometry-generic}
\end{figure*}

\section{Other Entanglement Measures}

\subsection{Max-Min entropies}
The min-entropy and max-entropy arise as two limiting cases of the Renyi-$k$ entropies, i.e., 
\begin{equation}
\begin{split}
 S_{max}(\rho_A) =& \textstyle\lim_{k\rightarrow 0} {\cal R}_A^k(\rho_A) = \log \left(\text{rank}~\rho_A\right)\\
  S_{min}(\rho_A) =& \textstyle\lim_{k\rightarrow \infty} {\cal R}_A^k(\rho_A) = -\log \left(\lambda_{max}(\rho_A)\right)
\end{split}
  \label{minmax}
\end{equation}
where $\lambda_{max}(\rho)$ is the largest eigenvalue of $\rho$. Given the three architectures in Figure~\ref{fig:circuit_diagram} and Section~\ref{sec:other}, their average max-min entanglement entropies of $50$ random circuit states are given in Figures~\ref{fig:pmax}~and~\ref{fig:pmin} as a function of $L$.

The max-entropy shows rapid initial growth caused by the circuit architecture in Figure~\ref{fig:circuit_diagram_left}, which increase the rank of the reduced density matrix $\rho_A$ every second layer until the rank saturates at the allowed maximum,  $2^{n_A}$. 
The non-negativeness and normalization of the reduced density matrix, i.e., $\lambda_i(\rho_A) \geq 0$ and $\sum_{i}\lambda_i(\rho_A) = 1$, then implies the change of the eigenvalue statistics from having one $\lambda$ being non-zero and having a value of $1$ to all $\lambda$'s being non-zero and having similar values around $2^{-n_A}$. Such decrease in the largest eigenvalue of $\rho_A$ is displayed in the min-entropy curve. We note this spectral change of $\rho_A$ can be a contributing factor for the emergence of the random matrix behavior of $\rho_A$, studied elsewhere \cite{KO2}.

\subsection{Geometric Measure}

Another way to measure the quantum entanglement of circuit  states is to study the geometric measure of entanglement, based on the overlap between the circuit state $|\psi(\theta)\rangle$ and its nearest product state \cite{Wei_2003}. It reads:
\begin{equation}
{\cal E}_g (|\psi(\theta)\rangle) = -\log\sup_{\alpha\in {\cal P}}~|\langle\alpha|\psi(\theta)\rangle|^2 \ ,
\label{geometric}
\end{equation}
where ${\cal P}$ is the set of qubit product states.
Figure~\ref{fig:geometry-generic} shows that the geometric measure of entanglement grows in the same pattern as of the entanglement entropy curve. However, note that the geometric measure is directly calculated from the full density matrix $\rho$, while the entanglement entropies are found from the reduced state $\rho_A$.

\providecommand{\href}[2]{#2}\begingroup\raggedright\endgroup


\end{document}